\documentclass[a4wide]{article}

\newtheorem{definition}{Definition}[section]
\newtheorem{prop}[definition]{Proposition}
\newtheorem{theorem}[definition]{Theorem}
\newtheorem{example}[definition]{Example}
\newtheorem{lemma}[definition]{Lemma}

\newenvironment{proof}[1]
{\begin{trivlist}\item[{\bf ~Proof}#1.
]}{\hfill $\Box$\end{trivlist}}

\usepackage{a4wide}

\newcommand{\Comp}{\circ}
\newcommand{\stt}[1]{\stackrel{#1}{\longrightarrow}}

\usepackage{microtype}

\usepackage{latexsym}
\usepackage{amsmath}
\usepackage{amssymb}
\usepackage{proof}
\usepackage{meaning}
\usepackage{bbm}
\allowdisplaybreaks

\usepackage{bm}

\usepackage{graphicx}

\usepackage{blkarray}

\usepackage[all]{xy}

\usepackage{wrapfig}

\usepackage{cutwin}

\usepackage{cmll}

\newcommand{\cut}{\operatorname{cut}}
\newcommand{\ax}{\textrm{ax}}

\newcommand{\MALL}{{\sf MALL}}
\newcommand{\MALLI}{{\sf MALL}(I)}
\newcommand{\LLI}{{\sf LL}(I)}

\newcommand{\MALLIC}{{\sf MALL}^{[\mathfrak{c}]}(I)}
\newcommand{\res}[2]{#1 \! \upharpoonright_{#2}}
\newcommand{\pf}[3]{#1_{[#2], #3}} 
\newcommand{\inpf}[4]{#1_{\, \vdash_{#2} [#3], \, #4} }
\newcommand{\inpfmean}[5]{ #1_{\, \vdash_{#2} ([#3],\, #4) \langle #5 \rangle}}
\newcommand{\abs}[1]{| #1 | }
\newcommand{\intf}[2]{| #1 | \times  |#2|}
\newcommand{\inseq}[3]{\vdash_{#1} [#2], \, #3}
\newcommand{\inseqmean}[4]{\vdash_{#1} ([#2], \, #3) \langle #4 \rangle }
\newcommand{\formmean}[2]{#1 \langle #2 \rangle }
\newcommand{\sblformmean}[2]{ #1  \langle \hspace{-.5ex} \langle #2 \rangle
\hspace{-.5ex}  \rangle }
\newcommand{\inseqmeansepare}[5]{
\vdash_{#1} [
\sblformmean{#2}{#4}
], \, \formmean{#3}{#5}}

\newcommand{\inseqmeanseparethree}[6]{\vdash_{#1} [
\sblformmean{#2}{#5}
], \, \formmean{#3}{#6^{'}}, 
\formmean{#4}{#6^{''}}}

\newcommand{\merge}[2]{#1 \raisebox{.5ex}{$\frown$} #2}

\newcommand{\Ex}[2]{
{\sf Ex}\left( #1, #2 \right)}
\newcommand{\ex}[2]{{\sf ex}\left( #1, #2 \right)}
\newcommand{\Excon}{{\sf Ex}}
\newcommand{\excon}{{\sf ex}}
\newcommand{\ptEx}[1]{{\sf Ex}\left( \sigma, #1
\right)}

\newcommand{\inEx}[2]{{\sf Ex}_{#2}\left( \sigma,  #1 \right)}

\newcommand{\Relc}{\Rel^{[\mathfrak{c}]}}

\newcommand{\MALLC}{{\sf MALL}^{[\mathfrak{c}]}}

\newcommand{\sbl}[1]{\mathfrak{sl}( #1 )}

\newcommand{\indrd}{\blacktriangleright^{\hspace{-.05cm}\textsc{\tiny
$I$}}}

\newcommand{\MLL}{{\sf MLL}}
\newcommand{\TR}[4]{{\sf Tr}^{#2}_{#3,#4}  \left( #1 \right)  }

\newcommand{\zr}[1]{\epsilon_{#1}}

\newcommand{\Rel}{{\sf Rel}}
\newcommand{\PInj}{{\sf PInj} }
\newcommand{\Pfn}{{\sf Pfn} }

\newcommand{\labelret}[1]{\rhd_{\hspace{-.3ex}#1}}
\newcommand{\labelcoret}[1]{\lhd_{\hspace{-.1ex}#1}}

\newcommand{\sblabs}[1]{[ #1 ] }

\title
{A MALL Geometry of Interaction
Based on Indexed Linear Logic}
\author
  {Masahiro 
   HAMANO \\
Institute of Information Science, Academia Sinica, Taiwan \\
\texttt{hamano@jaist.ac.jp}   }

\begin{document}

\label{firstpage}
\maketitle

\begin{abstract}
We construct a geometry of interaction (GoI: dynamic modeling of Gentzen-style
cut elimination) for multiplicative–-additive linear logic ($\MALL$) by
employing Bucciarelli--Ehrhard indexed linear logic $\MALL(I)$ to handle
 the additives.
 Our construction is an extension to the additives of
the Haghverdi--Scott categorical formulation (a multiplicative
GoI situation in a traced
 monoidal category)
for Girard's original GoI 1 \cite{Gi89}.
 The indices are shown to serve not only in their original denotational level, but also at a finer grained dynamic level
so that the peculiarities of additive cut elimination such as
 superposition, erasure of subproofs, and additive (co-) contraction can be handled
with the explicit use of indices.
Proofs are interpreted as indexed subsets in the category $\Rel$, but without
the explicit relational composition; instead,
 execution formulas are run
 pointwise on the interpretation
at each index, w.r.t symmetries of cuts, in a traced monoidal category
 with a reflexive object and
 a zero morphism.
The sets of indices diminish overall when an execution formula is run,
corresponding to
 the additive cut-elimination procedure (erasure),
and allowing recovery of the relational composition.
The main theorem is the invariance of the execution formulas
along cut elimination so that the formulas 
converge to the denotations of (cut-free) proofs.
\end{abstract}

\section{Introduction}
The indexed multiplicative additive linear logic
$\MALLI$, introduced by Bucciarelli--
Ehrhard \cite{BE}, is
a conservative extension of Girard's 
$\MALL$ in which all
formulas and proofs come equipped with sets of indices.
The usual $\MALL$ is stipulated to be the restriction
to the empty set.
The status of the indexed syntactical system is
noteworthy as it stems from the denotational semantics of Rel,
a simple, yet pivotal categorical model of $\MALL$.
With the enabling of an explicit notion of location in linear proof theory,
the indices can enumerate the locations of formulas and proofs,
corresponding to denotational interpretations of $\MALL$.
The notion of location becomes a requirement for the additives,
although it is redundant for the multipicatives, for which the singleton $\{ * \}$
suffices.
To work with \textit{parallelism}, which the additives bring intrinsically,
different locations need to be handled rather
than the sole location $*$.
In the context of parallelism, superpositions are known to typically arise
under the syntactic additive $\&$-rule.
Indices allow one to deal with superpositions
by identifying multiple occurrences of formulas in the different indices
and by enlarging (or restricting) the indices.

The original motivation for indexed logic was to provide a bridge
between a truth-valued semantics (for provability) for $\MALLI$
and the denotational semantics of (nonindexed)
$\MALL$. By means of this bridge, Bucciarelli--Ehrhard obtained a new kind of
denotational completeness theorem in \cite{BE} for $\MALL$ 
and later extended it to the exponential in \cite{BE2}.

This paper investigates indexed $\MALL$ from the perspective of
 a dynamic semantics for
 cut elimination, a topic that---to the best of our knowledge---has remained untouched
 aside from the precursory work of Duchesne \cite{Duchesne}
 since 
 the original work of Bucciarelli--Ehrhard \cite{BE}.
The dynamic semantics is 
the Girard project
of Geometry of Interaction (GoI), whereby cut-elimination
is modelled,
using operator algebras \cite{Gi89}
and more generally traced monoidal categories \cite{JSV96}.
The GoI project was successful \cite{Gi89,HS06} for $\MLL$ 
with the exponential, and inspired a new model of computation for $\beta$ reduction
of $\lambda$-calculus  \cite{DR}.
This paper aims to initiate
an exploration of 
how to combine
the two notions of \textit{location}, which the indexed logic brings,
and of \textit{dynamics}, which GoI brings to cut-elimination.
The combination is important in understanding additive
cut-elimination.
For this goal, we employ the indices to construct a GoI model
for (non-indexed) $\MALL$.
We combine the Haghverdi--Scott categorical GoI situation  \cite{HS06}
with the indices in such a way that  that the original  $\MLL$ GoI situation
represents a collapse to the singleton index $\{ * \}$.
The dynamics of cut-elimination
is captured
by a feedback mechanism determined by traces of morphisms interpreting proofs.
We further augment the situation with two kinds of actions,
identical and zero,
over the symmetries interpreting the cut-rule.
These two  actions provide representations of
matches and of mismatches among locations.
These come into play during
a Gentzen style cut-elimination procedure,
in which one encounters
noncommunication of individual proofs, 
due to the additive parallelism.
Crucial instances of GoI situation such as $\Rel_+$
and ${\sf Hilb_2}$ \cite{HS06} are directly lifted to our framework,
the latter of which is the operator algebraic origin
of the Girard project.

We study Girard's execution formula \cite{Gi89}
in the general categorical setting of a traced symmetric monoidal category.
The execution formula accommodates indices, and 
faithfully simulates $\MALL$ cut-elimination
by a hybrid method relating the indexed syntax to
the relational semantics. 
Each location in the relation interpreting a proof
is first assigned an endomorphism on a reflexive object $U$.
The cut-rule before execution is interpreted as a tensor product of two premise
morphisms, more loosely than their composition.
This interpretation allows extraction of the dynamical meaning
of the cut,
which the usual categorical composition hides 
by virtue of
its static approach. 
In the loose interpretation, there remain redundant indices
when interpreting rules:
however, they are shown to disappear,
while running the Execution formula in terms of the categorical trace
structure.
The disappearance of indices is modelled by 
zero morphisms, which exist in 
the traced monoidal categories for GoI.
Proof-theoretically,
the zero morphisms allow us to interpret
discarding subproofs specific to additive cut-elimination,
and in the way of theory of indices,
they provide
a way to to interpret mismatches among locations.
In traced monoidal categories, the zero morphisms are supposed 
 to  act partially on
symmetries for cut-formulas, and also to act partially
on retractions and co-retractions of the reflexive objects.
The latter action arises via tracing along the zero morphisms
which takes feedback into account along with the zero.
We prove zero-convergence which means that 
execution formulas converge to zero when two proofs interact
with mismatched locations.
Thus the execution formulas terminate to the denotational interpretations of
proofs, 
while diminishing sets of indices in order to
recover
the relational composition.
This is realized
by properly coupling indices
to trace axioms, especially for 
``generalized yanking''
and ``dinaturality''. 
The former axiom directly designates that
traces are primitive enough to retrieve the categorical
 composition in a monoidal category, and the latter axiom
concerns the interaction of bidirectional 
flow of morphisms.

In contrast to the precursory work of Duchesne \cite{Duchesne}
concerning both
indices and GoI,
the present paper accommodates the indices directly in GoI
semantics in order to simulate (nonindexed) $\MALL$ cut elimination.
The diminution of sets of indices is a typical dynamic aspect our GoI captures
using the zero morphisms of our traced monoidal category.
The precursor in \cite{Duchesne},
on the contrary, first accommodates the index into the
static category of $\Rel$, using the semantic method of localisation,
of which the indexed syntax provides a precise description.
Then (nonindexed) dynamic GoI action over the 
indexed denotational semantics is investigated
to characterise the static fix points as the denotational
interpretation.


We prove two main result:
(i) (Invariance of the execution formula during $\MALL$ normalization):  
The execution formula in our dynamic categorical modelling
is shown to converge to
the denotational interpretation of proofs
in the static categorical model. This characterises
the normalization of proofs by categorical
invariants.
(ii) (Diminution of indices while running the execution formula): 
The execution formula may converge to $0$, making the redundant indices
disappear.
Part (i) is seen as a pointwise collection of
invariants, as previously established
for the multiplicatives \cite{Gi89,HS06}.
Part (ii) is specific to the additives:
Proof-theoretically. it reflects erasure of subproofs as well as
additive (co) contraction and superposition, in 
cut-elimination. Category-theoretically, it ensures
that our categorical ingredient (the execution formula)
is fine grained enough to retrieve
a static monoidal category
as well as a relational category handling indices.

\bigskip

The rest of this paper is organized as follows:
Section 2 introduces a syntax $\MALLIC$
for indexed $\MALL$ with a cut list
as well as its relational counterpart
$\Relc$. A fundamental lemma is proved, which connects
a provable $\MALLIC$ sequent to 
an indexed subset of the interpretation of a $\MALL$ proof
with cuts.
In Section 3, $\MALL$ proof reduction for cut elimination
is lifted to $\MALLI$ proof transformation
with diminishing sets of indices.
Section 4 concerns our $\MALL$ GoI interpretation
by means of the indexed system
in a traced symmetric monoidal category with
zero morphism. 
Execution formulas are run indexwise, and the main theorem is
proved.
\section{$\MALLI$ with cut list and relational semantics}
\label{secsec}
\subsection{$\MALLI$ with cut list}
\noindent ({\bf Inference rules of $\MALLC$ with cut formulas}) \\
We accommodate a stack to record cut formulas
into the syntax of
the multiplicative–-additive linear logic $\MALL$.
To stress this, the system is written as $\MALLC$.
To accommodate the stack into the additive fragment,
one has to work with 
superpositions that arise inside the stack 
as well as in the conclusion (outside the stack).


A $\MALLC$ sequent $\vdash [\Delta], \Gamma$
with a cut list 
is a set $\Gamma$ of formula occurrences 
together with pairwise-dual formulas occurrences $\Delta$
inside the brackets.
Each dual pair in $\Delta$ is written $A, A^\perp$. 

\smallskip
\noindent Sequents are proved using the following rules:

\smallskip

\begin{minipage}[t]{1in}
$\infer[\ax]{\vdash A, A^\perp}{}$
\end{minipage}
\vspace{1ex}
\begin{minipage}[t]{2.2in}
$\infer[\otimes]{\vdash [\Delta_1, \Delta_2], \, \, \Gamma_1, \Gamma_2, A \otimes B}{
\vdash [\Delta_1], \, \, \Gamma_1, A & \vdash [\Delta_2], \, \, \Gamma_2, B}$
\end{minipage}
\vspace{1ex} 
\begin{minipage}[t]{1.2in}
$\infer[\parr]{\vdash [\Delta], \, \, \Gamma, A \parr B}{
\vdash [\Delta], \, \, \Gamma, A, B}$
\end{minipage}
\vspace{1ex} \\
\begin{minipage}[t]{2in}
  \hspace{16ex}
$\infer[\cut]{
\vdash [\Delta_1, \Delta_2, A, A^\perp],\, \, \Gamma_1, \Gamma_2}{\vdash
[\Delta_1], \, \, \Gamma_1,  A & 
\vdash
[\Delta_2], \, \, A^{\perp}, \Gamma_2}$
 \end{minipage}
\vspace{1.5ex}

\begin{minipage}[t]{2.3in}
$\infer[\&]{\vdash [\Delta_1, \Delta_2, \Sigma], \, \, \Gamma, A_1 \& A_2}{
\vdash [\Delta_1,\Sigma], \, \, \Gamma, A_1
& 
\vdash [\Delta_2, \Sigma], \, \, \Gamma, A_2
}$
\end{minipage}
\vspace{2ex}
\begin{minipage}[t]{1.4in}
$\infer[\oplus_1]{\vdash [\Delta], \, \, \Gamma, A_1 \oplus  A_2}{\vdash 
[\Delta], \, \,  \Gamma, A_1}$
\end{minipage}
\begin{minipage}[t]{1.4in}
$\infer[\oplus_2]{\vdash [\Delta], \, \, \Gamma, A_1 \oplus A_2}{
\vdash [\Delta], \, \, \Gamma, A_2}$
\end{minipage}

\smallskip

\noindent {\bf Note:} In the $\&$-rule, not only is $\Gamma$ superposed 
in the conclusion, but so is $\Sigma$ in the stack.
The superposition among cut formulas inside the stack
causes the well-known
additive (co-) contraction that arises in
$\MALL$ cut elimination.
The formulas occurrences $\Sigma$ is not deterministically chosen in the premises, so that
$\Sigma$ in general is neither empty (i.e., never superimpose cuts) nor
maximal (i.e., superimpose as many cuts as possible).
Thus, the $\&$-rule has several possible instances depending on the choice
of $\Sigma$.
The exchange rule is eliminated under the assumption that
formula occurrences are implicitly tracked by the premises and the 
conclusion of a rule.

\smallskip

We extend the above accommodation of cut lists
to Bucciarelli--Ehrhard indexed system $\MALLI$ \cite{BE}.
To stress this, the system is written
as $\MALLIC$.
The extension stipulates that
a set of indices is consistently associated
with each formula (including cut formulas) and sequent (including
cut lists).

We fix an index set $I$, once and for all. Each formula $A$ of
$\MALLI$ is associated with a set $d(A) \subseteq I$,
called the \textit{domain} of $A$.

\smallskip

\noindent 
({\bf $\MALLI$ formulas and domains}) \\ 
\noindent Formulas in the domain $J$ are defined by the following grammar:
$\bm{0}_\emptyset$ and $T_\emptyset$ are formulas of the domain
$\emptyset$. 
For any $J,K,L \subseteq I$ such that $K \cap L = \emptyset$
and $K \cup L = J$, \\
$\begin{aligned}
X_J  ::= \quad  \bm{1}_J \quad   \arrowvert \quad \bot_J \quad \arrowvert
\quad
X_J \otimes X_J \quad  \arrowvert 
\quad
X_J \parr X_J \quad 
\arrowvert
\quad  X_K \oplus X_L
\quad
\arrowvert \quad
X_K \& X_L .
\end{aligned}$

\noindent
For any $\MALLI$ formula $A$ with $d(A)=J$,
its negation $A^\perp$
with $d(A^\perp)=J$ is defined using the De Morgan duality for the $\MALL$
formula.

\smallskip

\noindent 
({\bf Restriction}) \\
For a $\MALLI$ formula $A$ with $d(A)=J$
and $K \subseteq J$, the restriction $\res{A}{K}$
of $A$ by $K$ is defined to be a $\MALLI$
formula in the domain $J \cap K$ as follows:

\noindent
$\res{\bm{0}_{\emptyset}}{K}= 
\bm{0}_{\emptyset}$ and $\res{T_{\emptyset}}{K}= T_{\emptyset}$.
\hspace{1ex} $\res{\perp_{J}}{K}= 
\perp_{J \cap K}$ and $\res{\bm{1}_{J}}{K}= 
\bm{1}_{J \cap K}$
\hspace{2ex} $\res{(A \otimes B)}{K}=
\res{A}{K} \! \otimes \; \res{B}{K}$ \\
$\res{(A \parr B)}{K}=
\res{A}{K} \! \parr \; \res{B}{K}$ 
\hspace{3ex}
$\res{(A \oplus B)}{K}=
\res{A}{K} \! \oplus \;  \res{B}{K}$
\hspace{5ex} $\res{(A \& B)}{K}=
\res{A}{K} \! \& \; \res{B}{K}$

It trivially follows that $\res{(A^\perp)}{K}=(\res{A}{K})^\perp$.
If $\Gamma$ is a sequence of $\MALLI$ formulas
$A_1, \ldots , A_n$
of domains $J$, we define
$\res{\Gamma}{K}=
\res{A_1}{K}, \ldots , \res{A_n}{K}$.

\smallskip

\noindent 
({\bf Inference rules of $\MALLIC$ with cut lists}) \\
We augment $\MALLC$ with indices.
This makes it possible to accommodate 
a stack recording cut formulas to the original $\MALLI$ \cite{BE}.
Although this is straightforward for the multiplicative fragment,
careful treatment is required for the listing of cut formulas
in the additive fragment.
Two kinds of sequences of formulas
are considered in our $\MALLIC$-syntax: One is a sequence $\Xi$ whose all formulas occurrences
have a same domain $J$ uniformly, which is denoted by $d(\Xi)=J$.
The other is a sequence $\Xi$ whose any formula occurrence
has a domain contained in $J$, which is denoted by $d(A) \subseteq J$.
Each sequent is of the form $\vdash_{J} [\Delta], \Gamma$,
in which $d(\Gamma)=J$ and $d(\Delta) \subseteq J$
with $d(A)=d(A^\bot) \subseteq J$ for any pairwise-dual formulas $A$ and
$A^\bot$ in $\Delta$ within the stack. 

\noindent {\bf Note: 
}
The uniformity requirement that all formulas in
$\Gamma$ have the same domain $I$ does not apply to the stack $\Delta$
of the cut formulas. Formulas from different cuts have various
domains contained in $J$.

\smallskip

\noindent 
{\bf Axioms and cut: }

\begin{minipage}[t]{.8in}
$\vdash_J \bm{1}_{J}$
\end{minipage}
\begin{minipage}[t]{.8in}
$\vdash_\emptyset \Gamma, T_\emptyset$
\end{minipage}
\begin{minipage}[t]{2.5in}
$\infer[\cut]{
\vdash_J  [\Delta_1, \Delta_2, A, A^\perp],\, \, \Gamma_1, \Gamma_2}{
\vdash_J
[\Delta_1], \, \, \Gamma_1,  A & 
\vdash_J
[\Delta_2], \, \, A^{\perp}, \Gamma_2}$
\end{minipage}
\begin{minipage}[b]{1.2in}
Note that $d(A)=d(A^\perp)=J$ for cut formulas $A$ and $A^\perp$.
\end{minipage} \\

\noindent {\bf Multiplicative rules:} \\

\begin{minipage}[t]{1.5in}
$\infer[\bot_J]{\vdash_{J} [\Delta], \, \, \Gamma, \bot_J}{
\vdash_{J} [\Delta], \, \, \Gamma}$
\end{minipage}
\begin{minipage}[t]{2.3in}
$\infer[\otimes]{\vdash_{J} [\Delta_1, \Delta_2 \,], \, \, \Gamma_1,
\Gamma_2, A \otimes B}{
\vdash_{J} [\Delta_1 ], \, \, \Gamma_1, A
& 
\vdash_{J} [\Delta_2], \, \, \Gamma_2, B
}$
\end{minipage}
\begin{minipage}[t]{2.3in}
$\infer[\parr]{\vdash_{J} [\Delta], \, \, \Gamma, A \parr B}{
\vdash_{J} [\Delta], \, \, \Gamma, A, B}$
\end{minipage} \\


\noindent {\bf Additive rules:} \\

\begin{minipage}[t]{2.8in}
$\infer[\&]{\vdash_{J_1 + J_2} [\Delta_1, \Delta_2, \Sigma \,], \, \, \Gamma, A_1 \& A_2}{
\vdash_{J_1} [\Delta_1,\res{\Sigma}{J_1}\, ], \, \, \res{\Gamma}{J_1}, A_1
& 
\vdash_{J_2} [\Delta_2, \res{\Sigma}{J_2} \, ], \, \, \res{\Gamma}{J_2}, A_2
}$
\end{minipage} \\
Note that the superposed context $\Gamma$
encompasses the \textit{whole} domain $J=J_1 + J_2$,
while the superposed context $\Sigma$ in the stack
has a domain \textit{contained in} $J$. \\

\smallskip

\noindent
\begin{minipage}[t]{1.8in}
$\infer[\oplus_1]{
\vdash_J  [\Delta],\, \, \Gamma, A_1 \oplus A_2}{\vdash_J
[\Delta], \, \, \Gamma,  A_1}$
\end{minipage}
\vspace{2ex}
\begin{minipage}[t]{2in}
$\infer[\oplus_2]{
\vdash_J  [\Delta],\, \, \Gamma, A_1 \oplus A_2}{\vdash_J
[\Delta], \, \, \Gamma, A_2}$
\end{minipage}
 \begin{minipage}[b]{1.7in} Note
  $d(A_{3-i})=\emptyset$ in \\ each rule $\oplus_i$ ($i=1,2$).
\end{minipage}

$\MALLI$
has no propositional variables; 
the only atomic formulas are the constants.
Then, the usual identity axiom is readily derived:

\begin{lemma}[Identity]
$\vdash_J A, A^\perp$
is provable for any $\MALLI$ formula $A$
of domain $J$.
\end{lemma}

\begin{lemma}[Restricting provable sequents] \label{respro}
If $\vdash_J [\Delta], \, \, \Gamma$ is provable, then so is \\
$\vdash_{J \cap K} [\res{\Delta}{K}],
  \, \, \res{\Gamma}{K}$
for any $K \subseteq I$.
\end{lemma}

For each inference rule of $\MALLIC$, if the conclusion sequent has the
 domain $\emptyset$, then so does the premise sequent(s). Thus, the rules
 for sequents deriving the empty domain are identified with the rules
 of $\MALLC$. As a consequence, every $\MALLIC$-proof $\pi$ for
 $\vdash_\emptyset [\Delta], \Gamma$ 
 contains only sequents of the empty domain. 
Hence $\pi$ is considered as a
$\MALL$-proof for $\vdash [\Delta],\Gamma$. To sum up,
\begin{lemma}
$\MALLIC$ is a conservative extension of $\MALLC$. 
\end{lemma}
Accordingly, in the sequel $\MALLC$ is identified with $\MALLC(\emptyset)$.

\subsection{$\MALLIC$ and Relational Semantics $\Relc$}
It is well known that the category $\Rel$
of sets and relations constitutes a denotational semantics of $\MALL$,
that is, the interpretation is invariant,
$(\pf{\pi}{\Delta}{\Gamma})^* =
(\pf{\pi'}{\Delta'}{\Gamma})^*$,
for any reduction 
$\pf{\pi}{\Delta}{\Gamma} \rhd \pf{\pi'}{\Delta'}{\Gamma}$
of $\MALL$ cut elimination.
In particular, the denotation of $\pi$ is equal 
to that of a cut-free $\pi'$ when $\Delta'$ is empty.
The cut rule is interpreted by \textit{relational composition} in
$\Rel$, and this interpretation makes the semantics denotational.

\begin{definition}[Denotational interpretation $(\pf{\pi}{\Delta}{\Gamma})^*$
in $\Rel$]\label{denotrel} 
Every $\MALL$ proof $\pf{\pi}{\Delta}{\Gamma}$
of a sequent $\vdash [\Delta] \, \, \Gamma$
is interpreted as a subset of an associated set of the conclusion
 (without
 the cut list),
\begin{align} \label{relhide}
(\pf{\pi}{\Delta}{\Gamma})^* \subseteq \abs{\Gamma}
\end{align}

\end{definition}
Note that in the interpretation, the cut formulas $\Delta$
become invisible by virtue of the relational composition:
Two relations $R_1 \subseteq G_1 \times D$ and 
$R_2 \subseteq D \times G_2$ compose in $\Rel$
\begin{align*}
R_2 \Comp R_1 = \{(g_1, g_2) \mid \exists \, d \in D \,  (g_1, d) \in R_1
\wedge (d,g_2) \in R_2 \} \subseteq G_1 \times G_2 
\end{align*}
The interpretation $\pi^*$ of (\ref{relhide}) is known
in \cite{BE} as the relational interpretation of $\MALL$ proofs in a 
compact closed category $\Rel_\times$ with biproducts $+$.
The interpretation is specified as follows
accordingly to the $\MALL$ rules, for which we refer
to the above $\MALLC$-rules ridden of the cut-lists.
First, every formula $A$ is interpreted as a set $\abs{A}$,
and every sequence $\Gamma=A_1, \ldots, A_n$
as $\abs{\Gamma}=\abs{A} \times \cdots
 \times \abs{A_n}$:
When $A$ is $A_1 \otimes A_2$ or $A_1 \wp A_2$, 
$\abs{A}=\abs{A_1} \times \abs{A_2}$,
and when $A$ is $A_1 \oplus A_2$ or $A_1 \& A_2$, 
$\abs{A}= (\{ 1\} \times \abs{A_1}) \cup (\{ 2 \} \times \abs{A_2})$. \\
Then 

\noindent (axiom) $\pi^* =  \{ (a,a) \mid a \in \abs{A}  \} \subseteq \abs{A,A^\perp}$

\noindent (cut-rule) $\pi^* =  \{ (\gamma_1, \gamma_2) \mid  (\gamma_1,
a) \in \pi_1^* \, \mbox{and} \,
(a, \gamma_2) \in \pi_2^* \} \subseteq \abs{\Gamma_1, \Gamma_2}$

\noindent ($\&$-rule)
$\pi^* =  \{ (\gamma,(1, a)) \mid  (\gamma, a) \in \pi_1^*  \}
\cup \{ (\gamma, (2, a) \mid  (\gamma, a) \in \pi_2^*  \}
\subseteq \abs{\Gamma, A_1 \& A_2}$

\noindent ($\oplus_i$-rule)
$\pi^* =  \{ (\gamma,(i, a)) \mid  (\gamma, a) \in \pi_i^*  \}
\subseteq \abs{\Gamma, A_1 \oplus A_2}$

\noindent ($\otimes$-rule)
$\pi^* =  \{ (\gamma_1, \gamma_2, (a_1, a_2)) \mid 
(\gamma_1, a_1) \in \pi_1^* \, \mbox{and} \, (\gamma_2, a_2) \in \pi_2^*
\} \subseteq \abs{\Gamma, A_1 \otimes A_2}$

\noindent ($\parr$-rule)
$\pi^* =  \{ (\gamma, (a_1, a_2)) \mid 
(\gamma, a_1, a_2 ) \in \pi^{'*} \} 
\subseteq \abs{\Gamma, A_1 \parr A_2}$

\smallskip

Our aim in this paper is to investigate a dynamics of cuts hidden in
such a static categorical composition. We begin by interpreting
proofs in $\Rel$ but without performing cuts by means of relational composition.
To stress this interpretation with the unexecuted cuts,
the categorical framework is denoted
by $\Relc$,
in which the cut list $[ \Delta ]$ is interpreted explicitly.

To deal with the additives in $\Relc$, we have to work with a
sublist and the set of
all the sublists:
Let $\Delta$ be $C_1, C_1^\bot, \ldots , C_m, C_m^\bot $.
For a subset $S$ of $\{ 1, \ldots, m \}$,
let $\Delta_S$ denote the sublist $ \ldots C_i, C_i^\bot \ldots$
where $i$ ranges in $S$.
Then the set $\sbl{\Delta}$ of all the sublists (including
$\Delta$ and $\Delta_\emptyset$) is defined;
\begin{align}
\sbl{\Delta} := \displaystyle
\{ \Delta_S \mid S \subseteq \{1, \ldots, m\}
   \} \label{unionDelta}
\end{align}

\noindent Consequently, we interpret (\ref{unionDelta})
as an object in $\Rel$ in terms of 
the disjoint union of each interpretations of the sublists:
\begin{align}
\sblabs{\,\sbl{\Delta}} = \! \! \! 
{\sum_{S \subseteq \{1, \ldots, m\}}{ \! \! \! \abs{\Delta_S}}},  \label{IntunionDelta}
\end{align}
\noindent in which $\abs{\Delta_S}$, for a nonempty sequence,
is the usual interpretation of the sequence
in $\Rel$ and $\abs{\Delta_\emptyset}:=\{ * \}=\abs{\bot}$.
The disjoint sum $\Sigma$ 
is taken over
different $S$'s.

In what follows in this paper, when $S$ is clear from context,
a sublist $\Delta_S$
is often abbreviated by $\widehat{\Delta}$, whose hat
indicates a pairwise deletion of some cut formulas.

\bigskip

\begin{lemma} \label{listdiv}
If $\Delta=\Delta_1, \ldots,  \Delta_k$ such that $\Delta_i$ are lists
of pairwise-dual formulas,
then 
\begin{align*}
\sblabs{\sbl{\Delta}} \cong \sblabs{\sbl{\Delta_1}} \times \cdots
  \times 
\sblabs{\sbl{\Delta_k}}
\end{align*}
E.g., a particular choice of each $\Delta_i$ is $A_i, A_i^\perp$.
\end{lemma}

In what follows in Sections \ref{secsec} and \ref{trdsec},
$\cong$ denotes an iso modulo the symmetry of the set-theoretical
cartesian product. In the sequel, the symmetry 
corresponds to the exchange of formula occurrences.
As the exchange is always clear from the context
and fixed,
we use the terminologies
$\subseteq_{\cong}$ and $\in_{\cong}$
consistently as follows:
$A \subseteq_{\cong} B$ (resp. $a \in_{\cong} B$)
means that $A$ is a subset (resp. a member) of $\sigma(B)$
where $\sigma$ is the exchange for $\cong$.

\begin{definition}[Interpretation $\abs{\pf{\pi}{\Delta}{\Gamma}}$
of proofs with unexecuted cuts
in $\Relc$] \label{intRelc}
Every $\MALL$ proof $\pf{\pi}{\Delta}{\Gamma}$
of a sequent $\vdash [\Delta],  \, \, \Gamma$
is interpreted by
$$
\abs{\pf{\pi}{\Delta}{\Gamma}} 
\subseteq \sblabs{\sbl{\Delta}} \times \abs{\Gamma},
$$

\noindent 
which is defined inductively and in the same manner as in Definition
 \ref{denotrel}, except for the cut rule to make the interpretation
 differ from the standard (\ref{relhide}) 
in that $\Delta$ is visible without performing the relational
composition.

\noindent (cut rule) \\
When $\pi$ is
$\infer[\cut]{
\vdash [\Delta_1, \Delta_2, A, A^\perp],\, \, \Gamma_1, \Gamma_2}{
\deduce{\vdash
[\Delta_1], \, \, \Gamma_1,  A}{\pi^1} & 
\vdash
\deduce{[\Delta_2], \, \, A^{\perp}, \Gamma_2}{\pi^2}}
 $
\begin{flalign*}
\abs{\pf{\pi}{\Delta_1,\Delta_2 A,  A^\perp}{\Gamma_1,\Gamma_2}}
: \cong
\abs{\pf{\pi^1}{\Delta_1}{\Gamma_1,A}}
\times
\abs{\pf{\pi^2}{\Delta_2 }{A^\perp, \Gamma_2}}  & \\
\subseteq \sblabs{\sbl{\Delta_1}} \times \abs{\Gamma_1} \times \abs{A}
\times
\sblabs{\sbl{\Delta_2}} \times \abs{A^\perp} \times \abs{\Gamma_2} &
 \subseteq_{\cong} \sblabs{\sbl{\Delta_1,\Delta_2}} \times \abs{A} \times \abs{A^\perp} \times \abs{\Gamma_1}
\times \abs{\Gamma_2}
\end{flalign*}
The symmetry used for the definition is the exchange 
between the conclusion of the cut
and merging those of $\pi_i$'s.
In obtaining the last inclusion, Lemma \ref{listdiv}
is used because the two lists $\Delta_1$ and $\Delta_2$ are disjoint.

\noindent ($\&$-rule) \\
When $\pi$ is
$\infer[\&]{\vdash [\Delta_1, \Delta_2, \Sigma \,], \, \, \Gamma, A_1 \& A_2}{
\deduce{\vdash  [\Delta_1, \Sigma \, ], \, \, \Gamma, A_1}{\pi^1}
& 
\deduce{\vdash [\Delta_2, \Sigma \, ], \, \, \Gamma, A_2}{\pi^2}
}$

\begin{flalign*}
& \abs{
\pf{\pi}{\Delta_1,\Delta_2, \Sigma}{\Gamma, A_1 \& A_2}}
:=  & \\
& \{\, (\lambda_1, (1,a_1)) \mid (\lambda_1, a_1) \in
\abs{
\pf{\pi^1}{\Delta_1, \Delta}{\Gamma, A_1}}
\, \} 
 +
\{ \, (\lambda_2, (2,a_2)) \mid (\lambda_2, a_2) \in
\abs{
\pf{\pi^2}{\Delta_2, \Sigma}{\Gamma, A_2}} 
\, \}  \subseteq  \\ &
 (\, \sblabs{\sbl{\Delta_1, \Sigma}} \! \times \! \abs{\Gamma} 
\! \times \! \{ 1 \} \times \abs{A_1} \, )
\! + \!
(\, \sblabs{\sbl{\Delta_2, \Sigma}} \! \times \! \abs{\Gamma} \! \times \! 
\{2 \} \times \abs{A_2} \, )
\subseteq \sblabs{\sbl{\Delta_1,\Delta_2,\Sigma}}
\times \abs{\Gamma} \! \times \! \abs{A_1 \& A_2}
\end{flalign*}
For the last inclusion, the \textit{monotonicity}, $\sblabs{\sbl{\Delta_i,\Sigma}}
\subseteq
\sblabs{\sbl{\Delta_1,\Delta_2, \Sigma}}$
is used.
\end{definition}

\smallskip

We extend Bucciarelli--Ehrhard
translation of Definition \ref{BETrans}
to Definition \ref{exBETrans}
to accommodate cut formulas inside the stack.

\begin{definition}[Translation of indexed relation $\gamma$
to $\MALLI$ 
sequent $\formmean{\Gamma}{\gamma}$ \cite{BE}] \label{BETrans} 
To any $\MALL$ formula $A$ and a family $a \in \abs{A}^J$,
a formula $\formmean{A}{a}$ of $\MALLI$ is
associated, with domain $J$ so that 
$\res{\formmean{A}{a}}{\emptyset}$ is $A$.
For a sequence  $\Gamma=A_1, \ldots, A_n$
of $\MALL$ formulas,
every $\gamma \in \abs{\Gamma}^J$
is written uniquely as $\gamma= \gamma^1 \times  \cdots 
\times \gamma^n$ with $\gamma^m \in \abs{A_m}^J$, and
we set $\formmean{\Gamma}{\gamma}=
\formmean{A_1}{\gamma_1}
, \ldots ,
\formmean{A_n}{\gamma_n}$.
\end{definition}
-For $A=\bm{0}$ or $A=T$, if $J \not  =\emptyset$,
we have $\abs{A}^J = \emptyset$, and
$\formmean{A}{a}$ is undefined. If
$J = \emptyset$, $\abs{A}^J$
has exactly one element, namely, the empty family $\emptyset$,
and we set $\formmean{\bm{0}}{\emptyset}=\bm{0}$ and
$\formmean{T}{\emptyset}=T$. \\
-If $A=\bm{1}$ or $A= \>\perp$, $a$ is the constant family,
and we set $\formmean{\bm{1}}{(*)_J}=\bm{1}_J$
and $\formmean{\perp}{(*)_J}=\perp_J$. \\
-If $A=B \otimes C$, then
$a=b \times c$ with
$b \in \abs{B}^J$
and $c \in \abs{C}^J$,
and we set
$\formmean{A}{a} = \formmean{B}{b} \otimes \formmean{C}{c} $
which is a well-formed formula of $\MALLI$ of domain $J$. 
Here $b \times c$ denotes the mediating morphism of the set-theoretical
cartesian product. \\ 
Similarly, for $A=B \parr C$, we set
$\formmean{A}{a} = \formmean{B}{b} \parr \formmean{C}{c} $. \\
-If $A=B \oplus C$,
then $a=b+c$
with $b \in \abs{B}^K$
and $c \in \abs{C}^L$ and $K+L=J$.
Then we set $\formmean{A}{a}=\formmean{B}{b} \oplus
\formmean{C}{c}$
which is a well-formed formula of $\MALLI$ of domain $J$. \\
Similarly for $A = B \; \& \; C$, we set
$\formmean{A}{a}=\formmean{B}{b} \, \& \, 
\formmean{C}{c}$.

\smallskip

\noindent ({\bf Notation}) \\
Let $X$ be a set and $J=J_1 + J_2$.
Every $x \in X^J$ yields the restrictions $x_i= \res{x}{J_i} \in
X^{J_i}$ with $i=1,2$.
Conversely, the two restrictions allow us to recover $x$. 
We write this as $x=\merge{x_1}{x_2}$.

\begin{definition}[Translation to $\MALLIC$
 sequent $\inseqmeansepare{J}{\Delta}{\Gamma}{\delta}{\gamma}$]
 \label{exBETrans} 
Let $\Delta$ be a sequence of pairwise-dual $\MALL$ formulas
and $\delta \in  \sblabs{\sbl{\Delta}}^J$ for some
$J \subseteq I$.
Then the $\MALLI$ sequence $\sblformmean{\Delta}{\delta}$
of pairwise-dual formulas is associated  
such that $d(\sblformmean{\Delta}{\delta})
 \subseteq J$ and
$\res{ \sblformmean{\Delta}{\delta}}{\emptyset}$ is $\Delta$ as follows.

First, we write $\Delta=\Delta_1, \ldots ,\Delta_n$, where
each $\Delta_i$ is a list of two dual formulas $A_i$ and $A_i^\bot$. 
By Lemma \ref{listdiv},
$\delta =\delta^1 \times \cdots \times \delta^n$,
so $\delta^i \in \sblabs{\sbl{\Delta_i}}^J$.
Because 
$\sbl{\Delta_i}=\{\Delta_i, \Delta_\emptyset \}$,
we have $\sblabs{\sbl{\Delta_i}}=\abs{A_i, A_i^\perp}  + \{ * \}$.
Recall that $\{ * \}$ interprets the empty list in (\ref{IntunionDelta}).
Thus every $\delta^i$ makes $J$ divide into $J=J_i + K_i$
to yield $\delta^i=\merge{\res{\delta^i}{J_i}\!}{\, \, 
\res{\delta^i}{K_i}}$
so that $\res{\delta^i}{J_i} \in \abs{A_i, A_i^\perp}^{J_i}$ and
 $\res{\delta^i}{K_i} \in \abs{*}^{K_i}$. (explicitly $J_i=\{ x \mid
\delta_i (x) \in \abs{A_i, A_i^\perp} \}$
and $K_i = \{ x \mid \delta_i(x) = * \}$.)
Then, using the $\res{\delta^i}{J_i}\,$, we define

{\centering 
$\begin{aligned}
\sblformmean{\Delta}{\delta}
= \formmean{\Delta_1}{\res{\delta^1}{J_1}}, \, \ldots,
\, \formmean{\Delta_n}{\res{\delta^n}{J_n}},
 \end{aligned}$
 \par}
 in which the two dual formulas in each
 $\formmean{\Delta_i}{\res{\delta^n}{J_i}}$
have the same domain $J_i \subseteq J$.






\smallskip

 Then, by employing Definition \ref{BETrans},
 for a given $\MALL$ sequent $\vdash [ \Delta],  \Gamma$,
every $\nu \in (\intf{\sbl{\Delta}}{\Gamma})^J$
is associated with a $\MALLI$ sequent, for which we write
$\nu = \delta \times \gamma$, so that
$\delta \in \sblabs{\sbl{\Delta}}^J$
and
$\gamma \in \abs{\Gamma}^J$: 
\begin{align}
\inseqmean{J}{\Delta}{\Gamma}{\nu}
 = 
{\inseqmeansepare{J}{\Delta}{\Gamma}{\delta}{\gamma} }\label{eqBETrans} 
\end{align}


\noindent
Here $\inseqmean{J}{\Delta}{\Gamma}{\nu}$ restricted to $\emptyset$
is $\vdash [\Delta], \Gamma$.
Note that all the formulas in $\formmean{\Gamma}{\gamma}$
have domain $J$, while each formula
in $\sblformmean{\Delta}{\delta}$
has a domain contained in $J$.
The $\sblformmean{\Delta}{\delta}$'s inside the stack
become a list of
pairwise-dual $\MALLI$ formulas in which each pair
has the same domain. 
\end{definition}

The translations commute with restriction of indices,
and Lemma \ref{respro} can be restated:
\begin{lemma}[Restricting translation]
\begin{itemize}
\item[-]
For any $\gamma \in \abs{\Gamma}^J$,
it holds that
$\formmean{\Gamma}{\res{\gamma}{J \cap K}}
=
\res{\formmean{\Gamma}{\gamma}}{K}
$.
\item[-]
For any $\delta \in \sblabs{\sbl{\Delta}}^J$,
it holds that
$\sblformmean{\Delta}{\res{\delta}{J \cap K}}
=
\res{\sblformmean{\Delta}{\delta}}{K}
$.
\item[-]
If $\vdash_J [\, \sblformmean{\Delta}{\delta} \,], \, \,
	\formmean{\Gamma}{\gamma}$ is provable, then so is
$\vdash_{J \cap K} [\, \res{\sblformmean{\Delta}{\delta}}{J \cap K} \,], \, \,
	\res{\formmean{\Gamma}{\gamma}}{J \cap K}$.

\end{itemize}
\end{lemma}

\subsection{Fundamental lemma}
Indexed linear logic arises essentially due to
its tight connection to the relational semantics.
The connection is
realized by a fundamental lemma 
due to Bucciarelli \& Ehrhard (proposition 20 of \cite{BE}) establishing a correspondence between
indexed sets in $\Rel$ and indexed sequents in $\MALLI$.
The former is semantic in $\MALL$, while the latter is syntactic in
$\MALLI$.
The fundamental lemma is shown to
be preserved under our extended
syntax and semantics, designed to accommodate cut formulas in
$\MALLIC$ and in $\Relc$, respectively.

\begin{prop}[Fundamental lemma \`{a} la Bucciarelli--Ehrhard] 
\label{FunLem}~~~ \\
For $\nu \in (\intf{\sbl{\Delta}}{\Gamma})^J$ with $J \subseteq I$,
the following two statements are equivalent and
induce a relationship $\res{\rho}{\emptyset}=\pi$ between
$\pi$ of (i) and $\rho$ of (ii):

\begin{enumerate}
\item[(i)]
There exists a $\MALLC$ proof $\pi$ such that
$$\nu \in \abs{\pf{\pi}{\Delta}{\Gamma}}^J.$$

\item[(ii)]
There exists a $\MALLIC$ proof $\rho$ of the sequent
$$\inseqmean{J}{\Delta}{\Gamma}{\nu}.$$

\end{enumerate}

\end{prop}

\begin{proof}{}
 See lemmas  \ref{iimpii} and \ref{iiimpi} in the Appendix
 \ref{apFunLem}.
\end{proof}

\section{Lifting $\MALL$ reduction over indices} \label{trdsec}
This section describes how our indexed syntax $\MALLIC$
analyzes Gentzen-style reduction of cut elimination
for nonindexed $\MALL$.
Every $\MALL$ reduction with cut elimination
is shown to be lifted to a directed transformation between two $\MALLI$ proofs.
These transformations diminish sets of the indices of proofs overall.

\begin{definition}[$\MALLIC$ proof transformation $\indrd$ with
 diminishing sets of indices]
 \label{defindrd}
A $\MALLIC$ transformation $\indrd$ with diminishing sets of indices,
written as $\inpf{\rho}{J}{\Delta}{\Gamma} \indrd 
\inpf{\rho'}{J'}{\Delta'}{\Gamma}$, is
a transformation from one $\MALLIC$ proof $\rho$ for $\inseq{J}{\Delta}{\Gamma}$
to another, $\rho'$ for $\inseq{J'}{\Delta'}{\Gamma}$ with $J' \subseteq J$,
satisfying the following condition:
\begin{description}
\item[-]
(Restriction to the empty domain) \\
Restricting the two $\MALLIC$ proofs to $\emptyset$ gives rise to a $\MALLC$ reduction
$\pf{\pi}{\Delta}{\Gamma} \, \rhd \, \pf{\pi'}{\Delta'}{\Gamma}$ by
cut-elimination. \\
Schematically, this can be written: \hspace{3ex}
$
\xymatrix{
\inpf{\rho}{J}{\Delta}{\Gamma}
\ar[d]_{\res{}{\emptyset}}
 & \indrd   & 
\inpf{\rho'}{J'}{\Delta'}{\Gamma}
\ar[d]^{\res{}{\emptyset}} \\
\pf{\pi}{\Delta}{\Gamma}
& 
 \rhd & 
\pf{\pi'}{\Delta'}{\Gamma}
}$


\end{description}
The transformation $\rho \, \indrd \, \rho'$ is called a
 \textit{lifting} of $\pi \, \rhd \, \pi'$.
The lifting simulates a $\MALL$ proof reduction for cut elimination
in terms of $\MALLI$ proof transformation.
\end{definition}

The lifting in Definition \ref{defindrd}
is not unique for a given $\MALLC$ reduction,
as any subset $J''$ of $J'$,
$\rho \, \indrd \, \res{\rho'}{J''}$ obviously becomes a lifting
for $\rho$ and $\rho'$ under this definition.

\smallskip

\begin{example} \label{case2}
The following is a $\MALLIC$ reduction with diminishing sets of indices
whose restriction to $\emptyset$ is a Gentzen reduction
eliminating the pairwise-dual additive connectives $\&$ and $\oplus$
in the cut formulas: \end{example}
$$\infer[\cut]{\vdash_{J_1 + J_2} [
\Delta_1, \Delta_2, \Omega, \Delta_3, (A_1 \& A_2),  (A_2^\bot \oplus A_1^\bot)], \Gamma, \Xi}
{
\infer[$\&$]{
\vdash_{J_1 + J_2} [\Delta_1, \Delta_2, \Omega],\Gamma, A_1 \& A_2 }
{\left\{
\begin{array}{c}
\infer*[\pi^i]{\vdash_{J_i} [\Delta_i, \res{\Omega}{J_i}], \res{\Gamma}{J_i}, A_i}{}
\end{array}
\right\}^{i=1,2}
}
& 
\infer*[\pi^3]{\infer[\oplus_1]{\vdash_{J_1 + J_2} [\Delta_3], \Xi , 
 A_1^\bot \oplus  A_2^\bot  }{
 \vdash_{J_1 + J_2} [\Delta_3],\Xi, A_1^\bot}}{}
 }$$
$$\deduce{\vspace{1.5ex}}{\indrd} \hspace{4cm} \infer[\cut]{
\vdash_{J_1} [
\Delta_1, \Omega,  \res{\Delta_3}{J_1}, A_1, \res{A_1^\bot}{J_1}], \Gamma, \Xi}{
\infer*[\pi^1]{\vdash_{J_1} [\Delta_1, \Omega], \Gamma, A_1}{}
& 
\infer*[\pi^3]{\vdash_{J_1} [\res{\Delta_3}{J_1}],\res{\Xi}{J_1},
\res{A_1^\bot}{J_1}}{}
}$$
The sets of the indices are diminished from $J_1 + J_2$ to $J_1$
as a result of erasing the subproof $\pi^2$
within the proof transformation.

\begin{prop}[Lifting to indexed transformation] \label{ltit}
Let $\nu \in \abs{\pf{\pi}{\Delta}{\Gamma}}^J$
and consider a $\MALL$ reduction $\pf{\pi}{\Delta}{\Gamma} \, \rhd \,
\pf{\pi'}{\Delta'}{\Gamma}$. Then there exist
$J' \subseteq J$ and
$\nu' \in \abs{\pf{\pi}{\Delta'}{\Gamma}}^{J'}$
lifting the given reduction:
\end{prop}
$$
\xymatrix{
\inpfmean{\rho}{J}{\Delta}{\Gamma}{\nu}
\ar[d]_{\res{}{\emptyset}}
 & \indrd   & 
\inpfmean{\rho'}{J'}{\Delta'}{\Gamma}{\nu'}
\ar[d]^{\res{}{\emptyset}} \\
\pf{\pi}{\Delta}{\Gamma}
& 
 \rhd & 
\pf{\pi'}{\Delta'}{\Gamma}
}$$
$\rho$ and $\rho'$ are $\MALL$ proofs ensured
by the fundamental lemma (Proposition \ref{FunLem})
for the sequents
$\inseqmean{J}{\Delta}{\Gamma}{\nu}$
and 
$\inseqmean{J'}{\Delta'}{\Gamma}{\nu'}$, respectively.
Hence,
we can also denote the lifting by

{\centering 
$\begin{aligned}
\nu \in \abs{\pf{\pi}{\Delta}{\Gamma}}^J
& \quad \indrd \quad &  \nu' \in \abs{\pf{\pi'}{\Delta'}{\Gamma}}^{J'}.
\end{aligned}$
\par}

\smallskip

\noindent Note: There is no straight connection between $\nu$
and $\nu'$  such as the former is the restriction to the latter.
\begin{proof}{}
For every kind of reduction $\rhd$, we can directly construct $\nu'$
together with $J'$. There are three crucial cases:

\noindent (Crucial case 1)
$$
\infer[\cut]{
\vdash [\Delta, B^\bot, A], B, \Gamma}{
\infer{\vdash B, B^\bot}{}
&
\infer*[\pi']{\vdash [\Delta], A, \Gamma}{}
}$$
(with $A$ and $B$ being different occurrences of the same formula)
reduces to 
$$\infer*[\pi']{\vdash [\Delta], B, \Gamma}{}$$
(identifying the occurrence of $A$ with $B$).

\smallskip

Let $\epsilon \in \abs{\ax_{B,B^\perp}}^J$
and $\tau \in \abs{\pf{\pi'}{\Delta}{A,\Gamma}}^J$.
Then, for each $j \in J$, we have $\epsilon_j=(b_j, b_j)$
with $b_j \in \abs{B}=\abs{B^\perp}$,
and $\tau_j=(\delta_j, a_j, \lambda_j)$
with $\delta_j \in \sblabs{\sbl{\Delta}}$,
$a_j \in \abs{A}$ and $\lambda_j \in \abs{\Gamma}$.
Note that $\nu_j=(\delta_j, b_j, a_j, b_j, \lambda_j)$.
We define
$$\begin{aligned}
J'=\{ j \in J \mid b_j = a_j\}  &  \quad  \text{and} & \nu'=\res{\tau}{J'}. 
\end{aligned}$$

\smallskip

\noindent (Crucial case 2) \\
This case is the $\MALL$ reduction arisen by
Example \ref{case2} above, when 
restricting to the empty domain $\emptyset$
and identifying ${\sf MALL}(\emptyset)$
with $\MALL$.

By $\pi$'s last rule, 
$\nu \cong \tau
 \times \lambda$ with $\tau \in \abs{\&(\pi^1,\pi^2)}^J$
and $\lambda \in \abs{\oplus_1 \! (\pi^3)}^J$,
where the conclusion of $\pi^i$  with $i=1,2$ is
$\vdash [\Delta_i, \Omega],
\Gamma, A_i$
and that of $\pi^3$ is
$\vdash [\Delta_3],\Xi, A_1^\bot$.
Then the $\&$-rule of the left premise divides $J$
 into $J=J_1 + J_2$. 
 We define 
\begin{align*}
& J'=J_1 \quad  \text{and} \quad   \nu'= \res{\tau}{J_1}
\! \times {\lambda'} \\
& \mbox{where $\lambda' \in \abs{\pi^3}^{J'}$ is defined by
$\lambda'_j := (x, a)$ if
$\lambda_j = (x,(1, a))$.
}\end{align*}

\noindent (Crucial case 3)
Here $\pi$
$$\infer[\cut]{ \vdash [\Xi, \Sigma_1, \Sigma_2,  \Omega, 
A, A^\bot], \Delta, \Gamma, B_1 \& B_2}
{\infer*[\rho]{\vdash [\Xi], \Delta, A}{}
& 
\infer[$\&$]{
\vdash [\Sigma_1, \Sigma_2, \Omega ], A^\bot, \Gamma, B_1 \& B_2 }
{
\left\{
\begin{array}{c}
\infer*[\pi^i]{\vdash [\Sigma_i, \Omega ], A^\bot, \Gamma, B_i}{}
\end{array}
\right\}^{i=1,2}
}
}$$ 
reduces to $\pi'$ 
$$\infer[\&]{
\vdash [
\Sigma_1, \Sigma_2, A, A^\bot, \, A, A^\bot,
\Xi, \Omega ], \Delta, \Gamma, B_1 \& B_2
}{
\left\{
\begin{array}{c}
\infer[\cut]{
\vdash [\Xi,  \Sigma_i, \Omega, A, A^\bot], \Delta, \Gamma, B_i
}{
\infer*[\rho]{\vdash [\Xi], \Delta, A}{}
&
\infer*[\pi^i]{\vdash [\Sigma_i, \Omega], A^\bot, \Gamma, B_i}{}
}
\end{array}
\right\}^{i=1,2}
}$$
Note that in the last $\&$-rule of $\pi'$,
$\Xi$ and $\Omega$ inside the stack are chosen to be superposed. \\
By $\pi$'s last rule, $\nu \cong \lambda \times \tau$,
so $\lambda \in \abs{\pf{\rho}{\Xi}{\Delta, A}}^J$
and $\tau \in \abs{\&(\pi^1,\pi^2)}^J$. The last $\&$-rule of the right
 premise divides $J$ into $J=J_1 + J_2$
so that $\tau \cong \merge{\tau_1}{\tau_2}$ and $\tau_i \in 
\abs{\pf{\pi^i}{\Sigma_i, \Omega}{A^\perp, \Gamma, B_i}}^{J_i}$.
Then $\res{\lambda}{J_i} \times \tau_i
\in_{\cong} \abs{\cut(\rho,\pi^i)}^{J_i}$. We define

{\centering 
$\begin{aligned}
J'=J & \quad  \text{and} & 
\nu' \cong
\merge{
(\res{\lambda}{J_1} \! \times \tau_1)}{
(\res{\lambda}{J_2} \! \times \tau_2)}.
\end{aligned}$
 \par}
\end{proof}

\section{$\MALL$ GoI Interpretation} \label{secGoIin}
\subsection{Execution formula with zero action on symmetries of cuts}
\subsubsection{Interpretation of indexed point
in $\MALL$ proof}~\\
Our categorical framework is a minimal part of the Haghverdi--Scott GoI
situation \cite{HS06} with a reflexive object $U$ in a 
traced symmetric monoidal category ${\cal C}$ with tensor unit $I$.
Ours in addition requires that ${\cal C}$ has zero morphisms,
in particular, a zero endomorphism $0_{U}$ on $U$:
\begin{align}
({\cal C}, \,  \otimes,\, I, \, s_{U,U}, \,   j: \! U \otimes U \lhd U
 \! :k, \,  0_{U}) \label{GoIsituzero}
\end{align}
\noindent $s_{U,U}$ is a symmetry $U \otimes U
\longrightarrow U \otimes U$ of tensor product.
$j: \! U^2 \lhd U \! : k$ denotes a pair of
morphisms $j$ and $k$ respectively from $U^2$ to $U$ and the other way
around. $j$ and $k$ are called respectively {\em co-retraction}
and {\em retraction}
for the reflexive $U$
when $k \Comp j=U \otimes U=U^2$.
The $m$-ary tensor folding
$\underbrace{\star \otimes \cdots \otimes \star}_{m}$
 is denoted by $\star^m$
both for object $\star$ or morphism $\star$.
The trace structure will be introduced later in (\ref{trfam}). 
 
We require the commutativity of the pair $(j,k)$ and the zero $0_U$;
\begin{align}
k \Comp 0_U \Comp j = 0_U \otimes 0_U \label{compz}
\end{align}
Indeed, (\ref{compz}) is equivalent to the two commutativity $j \Comp (0_U
\otimes \, 0_U) = 0_U \Comp j$ and $(0_U \otimes \, 0_U) \Comp k = k \Comp \,
0_U$.

\bigskip

{\bf Note:}
The zero morphism $0_U$ is absorbing with respect to composition, 
but not with respect to tensor. That is $f \otimes 0_{U}$
and $0_{U} \otimes f$
are not in general $0_{U^2}$ for any endomorphism
$f$ on $U$.
\begin{lemma}[tensoring zero] \label{tenzero}
$0_{U} \otimes 0_U = 0_{U^2}$. More generally,
$0_U^m = 0_{U^m}$ for any natural number $m$.
\end{lemma}
\begin{proof}{}
The first assertion is derived by
the condition (\ref{compz}). The general assertion is by iterating the
 condition $(k^n \otimes U) \Comp 
(k \Comp 0_U \Comp j) \Comp
(j^n \otimes U)= 0_U^{n+2}$.
\end{proof}

\smallskip

The zero endomorphism $0_{U}$
acts on the symmetry $s$ as follows.
\begin{definition}[zero-action $(s_{U,U})^0$] \label{zeroact}
$0$ action on the symmetry $s_{U,U}$ on $U^2$ 
is defined to {\em annihilate} the symmetry $s_{U,U}$ to the zero
 endomorphism on $U^2$.
\begin{align*}
(s_{U,U})^0 := 
 0_{U \otimes U}
\end{align*}
Alternatively, the action is defined to be the following 
decomposition 
in terms of conjugation (both precomposing and composing):
\begin{align*}
(s_{U,U})^0 
:=
(0_{U} \otimes  0_{U}) \, \Comp \, s_{U,U} \, \Comp \,  (0_{U} \otimes
 0_{U})
=0_{U} \otimes 0_{U}
= 0_{U \otimes U}
\end{align*}
   \end{definition}
We abbreviate $s_{U,U}$ and 
$(s_{U,U})^{0}$ 
as $s$ and $s^0$, respectively. 

To avoid collapsing the categorical framework
whose GoI interpretation becomes
the degenerate zero, we assume
\begin{align}
\text{$s$ is nonzero;} & \text{
that is, $s$ and $s^0$ are distinct endomorphisms on $U^2$
in ${\cal C}$.} \label{noncollapse}
\end{align}
This is a technical assumption for the main theorem
(Theorem \ref{mainthm}) of this paper
to characterise the diminution of the index set in terms of the 
convergence to the zero, distinguishable from the other morphisms.

\smallskip

The zero morphism, which is required in our framework, exists 
in crucial examples of GoI situations: 
 (i) $\Rel_+$ is $\Rel$ with the disjoint union $+$ of sets
as $\otimes$
and a reflexive object $\mathbb{N}$.
The \textit{empty relation} on $\mathbb{N}$ is the zero morphism.
Furthermore $0_{\mathbb{N}} + 0_{\mathbb{N}} = 0_{\mathbb{N} +
\mathbb{N}}$ sufficient to the condition (\ref{compz}).
(ii) The monoidal subcategories $\Pfn$ and $\PInj$ of $\Rel_+$,
both in which resides the zero morphism. $\PInj$ is known to be equivalent to the original
category ${\sf Hilb_2}$ of Hilbert spaces and partial isometries
for Girard's GoI 1 \cite{Gi89}. 

\noindent {\bf Note:} The above examples of GoI situations happen to be
sum-style monoidal structures \cite{HS11},
whose $\otimes$ is given by the disjoint union. 
The style is known to capture the notion of feedback 
as data flow in terms of streams of tokens around graphical networks.
However our categorical framework (\ref{GoIsituzero})
in the present paper is the general one,
hence does not assume that the monoidal product is sum-style.

\smallskip


\smallskip

The $i$-th constituent $x_i$ of 
$\bm{x}=(x_{1}, \ldots , x_{\ell})
\in \abs{\pf{\pi}{\Delta}{\Gamma}}$
of Definition \ref{intRelc}
corresponds, in terms of the membership relation,
to the $i$-th occurrence of formulas $\widehat{\Delta},\Gamma$
with a unique sublist $\widehat{\Delta}$.

\noindent

\begin{lemma}[Tag of $x_i$ with $\bm{x}=(x_1, \ldots ,  x_\ell)
\in \abs{\widehat{\Delta}} \times \abs{\Gamma}$
for $\bm{x}
\in \abs{\pf{\pi}{\Delta}{\Gamma}}$]
Every element $\bm{x}=(x_1, \ldots ,  x_\ell)
\in \abs{\pf{\pi}{\Delta}{\Gamma}}$ interpreting $\pi$
in Definition \ref{intRelc}
belongs to $\abs{\widehat{\Delta}} \times \abs{\Gamma}$
with a unique sublist $\widehat{\Delta}$ of $\Delta$.
That is, for $\bm{x}=(x_1, \ldots ,  x_\ell)$, there exists 
a unique sublist $\widehat{\Delta}$ such that 
the $i$-th constituent $x_i \in \abs{A_i}$ for the $i$-th formula $A_i$ in $\widehat{\Delta}$
(resp. in $\Gamma$) when $i \leq 2m$ (resp. $i > 2m$), where
$2m$ is the number of formulas in $\widehat{\Delta}$.
The formula $A_i$ is called the {\em tag} of the $i$-th constituent
$x_i$ of $\bm{x}$ in
 $\abs{\pi}$.  Note the sublist
 $\widehat{\Delta}$ is determined not only by $\pi$ but also by $\bm{x}$, as shown in the
 construction (\ref{conschedelta})
in the following proof.
\end{lemma}
\begin{proof}{}
As $\abs{\pf{\pi}{\Delta}{\Gamma}} \cong
\abs{\pf{\pi}{\Delta}{\Gamma}}^{ \{ * \}}$,
by (\ref{eqBETrans}) in Definition \ref{exBETrans}
when $J$ is the singleton set $\{ * \}$, every
$\bm{x}=(x_1, \ldots ,  x_\ell)$ factors $\bm{x} = \bm{x}' \times \bm{x}''$
so that 
$$\begin{aligned}
\vdash_{\{ * \}} \formmean{([\Delta], \Gamma)}{\bm{x}}
= {\inseqmeansepare{\{*\}}{\Delta}{\Gamma}{\bm{x}'}{\bm{x}''}}
\end{aligned}$$
While all the formulas in the sequence 
$\formmean{\Gamma}{\bm{x}''}$ of $\MALLI$ formulas
has the domain $\{ *\}$,
each formula in the sequence $\sblformmean{\Delta}{\bm{x}'}$
has the domain either $\{ * \}$ or $\emptyset$. Thus the
 unique sublist $\widehat{\Delta}$ is determined by the following two
steps: (i) Ridding $\sblformmean{\Delta}{\bm{x}'}$
of all the formulas $D$ such that $d(D)=\emptyset$,
which yields the subsequence of $\sblformmean{\Delta}{\bm{x}'}$.
(ii) $\widehat{\Delta}$ is defined to be the
subsequence of (i) restricted to $\emptyset$ (i.e., forgetting the domain)
in order to obtain non-indexed formulas.
\noindent 
\begin{align} \label{conschedelta}
\text{\parbox{.85\textwidth}
{To be short for (i) and (ii), 
the sublist $\widehat{\Delta}$ is
the unique one making $\sblformmean{\widehat{\Delta}}{\bm{x}'}$
agree with $\sblformmean{\Delta}{\bm{x}'}$ ridden of all the formulas
 whose domain is $\emptyset$.}}
\end{align}
Note whenever a cut formula occurs in $\widehat{\Delta}$,
so does its dual formula, hence
$\bm{x}'=(x_1, \ldots ,  x_{2m})$ for a natural number $m$,
then $\bm{x}''=(x_{2m+1}, \ldots ,  x_\ell)$.
By Definitions \ref{BETrans} and \ref{exBETrans},
the membership relation for the assertion of $x_i$
and the $i$-th formula in $\widehat{\Delta}, \Gamma$
follows.
\end{proof}

In what follows, 
we make permutations
$(x_{\tau(1)}, \ldots , x_{\tau(\ell)})$
 among the constituents
implicit so that $\bm{x}$ is up to the permutation.
This is because the permutation corresponds to the exchange rule
eliminated from our syntax.
The permutations will be reflected by the symmetry of monoidal product
of ${\cal C}$,
in the following interpretation $\Mean{\bm{x}}$, which we denote by $\cong$.

\begin{definition}\label{endopartia}{\bf Endomorphism $\Mean{\bm{x}}$ on tensor folding 
$U$'s 
and tensor folding $\sigma_{\bm{x}}$ of symmetry $s$ and of zero $s^0$}
 \begin{itemize}
 \item 
Every $\bm{x}=(x_1, \ldots ,  x_\ell)
\in \abs{\pf{\pi}{\Delta}{\Gamma}}$
is interpreted as
 an endomorphism $\Mean{\bm{x}}_\pi$
on the tensor product $U^\ell$
 together with an endomorphism $\sigma_{\bm{x}}^\pi$
on a subfactor 
$U^{2m}$ of $U^\ell$. The endomorphism $\sigma_{\bm{x}}^\pi$ interprets cut rules in
$\pf{\pi}{\Delta}{\Gamma}$ and is 
an $m$-ary tensor folding of morphisms which are either
$s^1=s$ or $s^0=0$ 
(cf. Definition \ref{zeroact}) both on $U^2$:
\begin{align}
 \Mean{\bm{x}}_\pi: U^\ell \longrightarrow U^\ell & &  \mbox{and} & &
\sigma_{\bm{x}}^\pi= \otimes_{i=1}^{m} \; s^{\eta(i)} & & 
\mbox{where $\eta_{\bm{x}}$ is a $\{0,1 \}$-valued function.} \label{sigma}
\end{align}
To distinguish each $i$-th component $U$ of $U^\ell$,
we label each component $U_{x_i}$ with $x_i$ by abuse 
of notation, since the label $x_i$
is always clearly specified to designate the $i$-th constituent of $\bm{x}$.
Then $x_i \in \abs{A_i}$ so that $A_i$ is the tag of $x_i$.
Under this labelling, $\eta(i)$ of $s^{\eta(i)}$
on $U_a \otimes U_{a'}$ is defined to be 
the Kronecker delta
$\delta_{a,a'}$,
where the tags of $a$ and $a'$ are pairwise dual formulas in
$\widehat{\Delta}$. That is, 
\begin{align}
 \sigma_{\bm{x}}^\pi = \otimes 
 \; s^{\delta_{a,a'}} &  & \mbox{where 
$a$ and $a'$ range 
so that their tags are pairwise
} \label{spsigma} \\ & &  
\mbox{dual cut formulas in $\widehat{\Delta}$.
} \nonumber
\end{align}
We define $(\Mean{\bm{x}}_\pi, \sigma_{\bm{x}}^\pi)$
by induction on the construction of the proof $\pi$. 
\item
We simultaneously define that a component 
$U_{x_i}$ such that the tag of $x_i$ is a formula in $\Gamma$
is {\em contracted} by the induction on $\pi$
\footnote{ For the choice of $x_i$ (i.e., a choice of a formula (not in
the cut-list $\Delta$ but) in $\Gamma$ ),
the construction is free from cut.}.
Since $U_{x_i}$ appears 
both in the co-domain and in the domain
of $\Mean{\bm{x}}$,
every contracted component 
in the domain (resp. co-domain) of $\Mean{\bm{x}}$
is a domain (resp. co-domain) of a unique retraction
(resp. co-retraction), called {\em
associated retraction} (resp. {\em associated co-retraction})
\footnote{assoc-ret (resp. assoc-coret) for short}.
We simply say $i$ is {\em contracted} when so is $U_{x_i}$.
\end{itemize}
\end{definition}
In the definition, $\pi_1$ and $\pi_2$
denote the two premise proofs of the binary rules,
and $\pi'$ of the unary rules.

\noindent (Axiom) \\
$\bm{x}=(\bar{*},*) \in \abs{A^\perp, A}$ with $A=1$ and $A^\perp =
\perp$. 
We define $\Mean{\bm{x}}_\pi$ to be a symmetry 
$s_{U_{\bar{*}},U_*}$ on $U_{\bar{*}} \otimes U_*$ of ${\cal C}$.
Because $\pi$ is cut-free, $\sigma_{\bm{x}}^\pi$ is empty by definition.

\noindent $\bm{x}$ has no contracted component
 so that neither $U_*$ nor $U_{\bar{*}}$ are contracted.

\smallskip

\noindent(Cut rule) \\
$\bm{x}= (\bm{v}^1, \bm{v}^2, a, a',\bm{w}^1  ,\bm{w}^2 )
\in \abs{\widehat{\Delta_1}} \times
\abs{\widehat{\Delta_2}} \times \abs{A} \times \abs{A^\perp}
\times \abs{\Gamma_1} \times \abs{\Gamma_2}
$, so $\bm{x}_1= (\bm{v}^1,\bm{w}^1,a)$ and 
$\bm{x}_2=  (\bm{v}^2,
 a', \bm{w}^2 )$ belong respectively to
$\abs{\pi_1}$ and to $\abs{\pi_2}$. \\
We define 
 \begin{align*}
&  \Mean{\bm{x}}_\pi \cong \Mean{\bm{x}_1}_{\pi_1} \otimes 
\Mean{\bm{x}_2}_{\pi_2} \quad
  \text{and}  
&  \sigma_{\bm{x}}^\pi \text{~is~}
  \sigma_{\bm{x}_1}^{\pi_1} \otimes \sigma_{\bm{x}_2}^{\pi_2} \otimes 
(s_{U_a,U_{a'}})^{\delta_{a,a'}}
  \end{align*}
That is, 
if 
$a=a'$ (resp. $a \not =a'$),  then 
$\sigma_{\bm{x}}^\pi$ on $U_a \otimes U_{a'}$
is $s$ (resp. $s^0$), and 
$\sigma_{\bm{x}}^\pi$ on the remaining components is
$\sigma_{\bm{x}_1} \otimes \sigma_{\bm{x}_2}$ .
Note the definition makes sense because
$\sigma_{\bm{x}_1}^{\pi_1} \otimes \sigma_{\bm{x}_2}^{\pi_2}$ acts on the domain
distinct both from $U_{a}$ and $U_{a'}$.

We say the cut (of the last rule of $\pi$) {\em matches} (resp. {\em mismatches})
in $\bm{x}$ if $a = a'$ (resp. otherwise).

\smallskip

\noindent($\parr$-rule) \\
$\bm{x}=(\bm{v},(a,b))$,
so that $\bm{x'}=(\bm{v}, a, b) \in \abs{\pi'}$.
Note $A \parr B$ is the tag of $(a,b)$ in $\pi$, while $A$ (resp. $B$)
is the tag of $a$ (resp. $b$) in the premise.
$\Mean{\bm{x}}_\pi$ is
obtained directly from
$\Mean{\bm{x'}}_{\pi'}$ on $U^{\ell +1}= U^\ell \times U_{(a,b)}$
by the retraction $U_{a} \otimes U_{b} \lhd
U_{(a,b)}$.
That is,
$\Mean{\bm{x}}_\pi = \Mean{\bm{x'}}_{\pi'}^{(j,k)}=(U^{\ell -1} \otimes j)
\, \Comp \, \Mean{\bm{x'}}_{\pi'} \, \Comp \, (U^{\ell -1} \otimes k)$.
We also define $\sigma_{\bm{x}}^\pi$ by $\sigma_{\bm{x'}}^{\pi'}$.  

\smallskip

\noindent($\otimes$-rule) \\
$\bm{x}= (\bm{v}^1, \bm{v}^2, \bm{w}^1, \bm{w}^2,
(a, b))
\in \abs{\widehat{\Delta_1}} \times
\abs{\widehat{\Delta_2}} 
\times \abs{\Gamma_1} \times \abs{\Gamma_2}
\times \abs{A} \times \abs{B}$,
so that $\bm{x}_1= (\bm{v}^1,\bm{w}^1, a)$ and 
$\bm{x}_2 = (\bm{v}^2, \bm{w}^2 , b)$ are respectively from
$\abs{\pi_1}$ and $\abs{\pi_2}$.
Note $A \otimes B$ is the tag of $(a,b)$ in $\pi$, while $A$ (resp. $B$)
is the tag of $a$ (resp. $b$) in $\pi_1$ (resp. in $\pi_2$).
The endomorphism $\Mean{\bm{x}}$ is obtained directly from
$\Mean{\bm{x}_1}_{\pi_1} \otimes \Mean{\bm{x}_2}_{\pi_2}$ on $U^{\ell+1}
= U^\ell \times U_{(a,b)}$
by the retraction $U_{a} \otimes U_{b} \lhd
U_{(a, b)}$.
That is,
$\Mean{\bm{x}}_\pi \cong (\Mean{\bm{x}_1}_{\pi_1}  
\otimes \Mean{\bm{x}_2}_{\pi_2})^{(j,k)}
=(U^{\ell -1} \otimes j)
\, \Comp \, (\Mean{\bm{x}_1}_{\pi_1}  \otimes \Mean{\bm{x}_2}_{\pi_2})
\, \Comp \, (U^{\ell -1} \otimes k)$. We also define $\sigma_{\bm{x}}^\pi$ by
$\sigma_{\bm{x}_1}^{\pi_1} \otimes \sigma_{\bm{x}_2}^{\pi_2}$.

\bigskip
 
In the above both multiplicatives rules ($\otimes$ and $\parr$),
the introduced  $U_{(a,b)}$ in the domain (resp. co-domain) is 
a contracted component, and the assoc-ret
 (resp.  assoc-coret) is $k :  U_{(a,b)} \rhd U_a \otimes U_b$ 
 (resp. $j : U_a \otimes U_b \lhd U_{(a,b)}$).
 Other contracted
   components are those of
$\Mean{(\bm{v},a,b)}$ for $\parr$
and $\Mean{(\bm{v}_1,\bm{w}_1,a)}$ 
and $\Mean{(\bm{v}_2,\bm{w}_2,b)}$ for $\otimes$
distinct from the components $U_{a}$ and
 $U_{b}$. 
Note that $(\bm{v},a,b) \in \abs{\mbox{premise of $\parr$}}$
and $(\bm{v}_1,\bm{w}_1,a) \in
\abs{\mbox{left premise of $\otimes$}}$
and 
$(\bm{v}_2,\bm{w}_2,b) \in
\abs{\mbox{right premise of $\otimes$}}$.

\smallskip

\noindent($\&$-rule) \\
$\bm{x}$ is either $(\bm{v},(1, a))$
or $(\bm{v},(2, a))$, so that $(\bm{v}, a )$ are
either from  $\abs{\pi_1}$ or
$\abs{\pi_2}$,
 respectively.
Note $A_1 \& A_2$ is the tag of $(i,a)$ in $\pi$, while $A_i$
is the tag of $a$ in $\pi_1$ or in $\pi_2$ when $i=1$
or $i=2$, respectively.
We define $\Mean{\bm{x}}_\pi=\Mean{(\bm{v}, a)}_{\pi_i}$
by relabelling the component $U_a$ either by $U_{(1,a)}$ or $U_{(2,a)}$
for the domain (equally for the codomain) of $\Mean{\bm{x}}_\pi$.
We also define $\sigma_{\bm{x}}^\pi$ by $\sigma_{(\bm{v}, a)}^{\pi_i}$.

\smallskip

\noindent($\oplus_i$-rule) 
Same as $\&$-rule but using the unique premise $\pi'$ deterministically.

\bigskip

In the above both additives rules ($\&$ and $\oplus_i$), 
contracted components are those of $\Mean{(\bm{v}, a)}$ under the
 relabelling $U_{a}$ by $U_{(i, a)}$ for the component
of the domain (equally of the codomain).
Note that $(\bm{v}, a)$ belongs to one of 
$\abs{\mbox{left premise}}$ and 
$\abs{\mbox{right premise}}$ in $\&$-rule depending on
$i=1$ or $i=2$, 
and obviously to $\abs{\mbox{the unique premise}}$ in $\oplus_i$-rule.

\bigskip

In the sequel, 
the pair of Definition \ref{endopartia}
is simply written  
$(\Mean{\bm{x}},\sigma_{\bm{x}})$ 
by omitting $\pi$,  
since the proof $\pi$ will be
always specified clearly from the context.

\smallskip

\noindent({\bf Remark on Def \ref{endopartia} }) 
[{\bf The endomorphism $\Mean{\bm{x}}$ as an I/O box}] \\
The endomorphism $\Mean{\bm{x}}$ is seen as an input/output (I/O) box
on the $(n+2m)$-ary tensor folding of $U$,
whose inputs/outputs are the formulas occurring in $\Gamma, \widehat{\Delta}$,
in which $\Gamma$ contains $n$ occurrences of formulas,
and a sublist $\widehat{\Delta}$
contains $2m$ occurrences of (pairwise dual) formulas.
The formulas are the tags of $x_i$s where $\bm{x}=(x_1, \ldots ,
      x_\ell) \in \abs{\pf{\pi}{\Delta}{\Gamma}}$. \\
The endomorphism $\sigma_{\bm{x}}$
is seen as a more special box
consisting of $m$-ary tensor folding of $\{ s, s^0\}$
for the I/O formulas in the sublist $\widehat{\Delta}$. 
See Figure \ref{firstfigs} below for $(\Mean{\bm{x}}, \sigma_{\bm{x}})$.

\smallskip

A characterisation of contracted component is derived:
\begin{lemma} \label{lemtag}
$U_{x_i}$ is a contracted component of $\Mean{\bm{x}}$ for $\bm{x} \in
 \Mean{\pi}$
if and only if $x_i$'s tag
contains a multiplicative connective (i.e., $\otimes$ or $\parr$).
\end{lemma}
\begin{proof}{}
Straightforward accordingly to the inductive step of Definition
\ref{endopartia}, in which
the retraction and the co-retraction are used
only for multiplicatives-rules ($\parr$ or $\otimes$)
so that $\Mean{\bm{x}}$ is constructed 
$\Mean{\bm{x'}}^{(j,k)}$ or
$(\Mean{\bm{x}_1}  \otimes \Mean{\bm{x}_2})^{(j,k)}$,
respectively.
\end{proof}


\noindent({\bf
labelling associated retractions
$\labelret{i}$ and co-retractions $\labelcoret{i}$}) \\
Every associated retraction $\rhd$ (resp. associated co-retraction
$\lhd$) is by Definition \ref{endopartia}
uniquely labelled
with a contracted $i$ such that the tag of
$x_i$ is a formula in $\Gamma$.
The labelling is written
$\labelret{i}$ (resp. $\labelcoret{i}$).
In what follows, the labelling is made implicit
except when an explicit labelling makes an explanation
easier to understand.

\smallskip

\noindent ({\bf assoc-rets and assoc-corets in I/O box $\Mean{\bm{x}}$})  \\
When the endomorphism $\Mean{\bm{x}}$ is seen as the I/O box,
the assoc-rets and the assoc-corets are those $\rhd$'s and $\lhd$'s
whose domains and co-domains lie  
respectively among the inputs and the outputs of $\Mean{\bm{x}}$.
By the construction of $\Mean{\bm{x}}$, they lie pairwise
in the inputs and the outputs. See Figure  \ref{firstfigs} 
for $\Mean{\bm{x}}$ depicting the occurrence of the assoc-rets $\rhd$'s and
the assoc-corets $\lhd$'s.

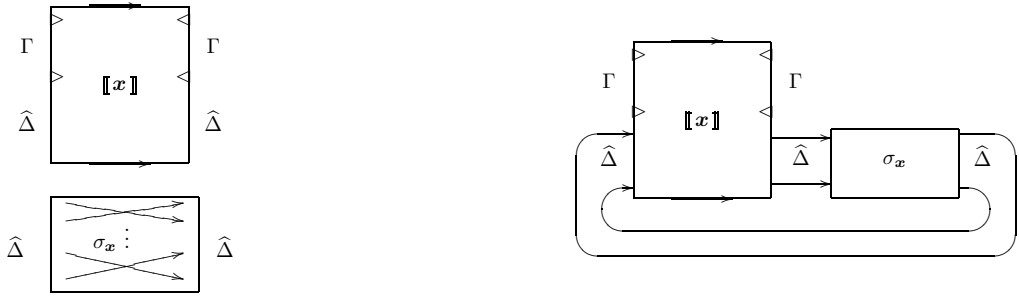
\begin{figure}[!htbp]
$\begin{array}{c}
\xymatrix@W=1pc @H=.1pc @R=.1pc{  &  \hspace{-3ex}   \rhd \ar@<1.5ex>[r]
& 
\hspace{3.2ex}  \lhd \\
 \Gamma \hspace{-10ex} \vspace{5ex}  &  &  & \hspace{-12ex} \Gamma     \\
 &  \hspace{-3ex} \rhd  &   \hspace{3.2ex} \lhd \\
  &  &     \\ 
\widehat{\Delta} \hspace{-10ex} \vspace{5ex}  & &  & \hspace{-12ex} \widehat{\Delta}  \\
& \ar@<-1.3ex>[r]   &  \ar@{}[uuuuul]|(.5){\txt{\normalsize $\Mean{\bm{x}}$}}
\save"1,2"."6,3"*[F]\frm{}\restore  } 
\end{array}$
$\begin{array}{c}
\vspace{-3ex}
\xymatrix@W=.1pc @H=.1pc @R=.1pc{
 &     \ar[drr]  &   &   \\
  &       \ar[urr]  &    &  &    \\
\widehat{\Delta} \hspace{-5ex}   & \ar[drr]  &   & &   \hspace{-7ex}
 \widehat{\Delta}  \\
 &       \ar[urr]  &    & \ar@{}[uuull]|(.6){\txt{\normalsize$\sigma_{\bm{x}}$ $\vdots$}}
\save"1,2"."4,4"*[F]\frm{}\restore    
}
\end{array}$ \\
\begin{tabular}{cc}
$\Mean{\bm{x}}$ with assoc-rets $\rhd$
  and assoc-corets $\lhd$
& 
$\sigma_{\bm{x}}$
whose each cross is either $s$ or $s^0$
\end{tabular}
\caption{$(\Mean{\bm{x}}, \sigma_{\bm{x}})$ of Definition
 \ref{endopartia}
and assoc-rets
  and assoc-corets of Definition \ref{endopartia} }
\label{firstfigs}
\end{figure}

\noindent ({\bf Convention omitting $\operatorname{Id}_U$s})
When an indicated occurrence of a contracted component $U$ is
clear from the context,
$\rhd \Comp \Mean{\bm{x}}$ (resp. $\Mean{\bm{x}} \Comp \lhd$) 
is an abbreviation for the composition 
~$(\operatorname{Id}_U \otimes \cdots \otimes \operatorname{Id}_U \otimes \rhd \otimes \operatorname{Id}_U \otimes \cdots \otimes \operatorname{Id}_U )
 \Comp \Mean{\bm{x}}$
(resp. $\Mean{\bm{x}} \Comp
(\operatorname{Id}_U \otimes \cdots \otimes \operatorname{Id}_U \otimes \lhd \otimes \operatorname{Id}_U \otimes \cdots \otimes \operatorname{Id}_U )
$, where the domain of
$\rhd$ (resp. the codomain of $\lhd$) is the
contracted $U$.
This abbreviation is generalised for plural
indicated occurrences of
contracted components $U_1, \ldots U_r$
in $U^\ell$ as follows:
$ (\otimes^{r} \rhd)
\Comp \Mean{\bm{x}}$
(resp. $\Mean{\bm{x}} \Comp (\otimes^{r} \lhd )$)
stands for
the composition (resp. precomposition) to
$\Mean{\bm{x}}$
by the morphism tensoring $\rhd$ (resp. $\lhd$)
on the contracted components indicated
and $\operatorname{Id}_U$ on the remaining components.
Note because $r \leq \ell$, the abbreviation is for omitting identities on
$U$s. 
\renewcommand\windowpagestuff{%
\flushleft
\vspace{1ex}
\scalebox{.8}{ 
\begin{minipage}{0.3\hsize}
\xymatrix@C=1pc@R=.5pc
{ \ar@{}[ddr] |{\txt{\normalsize $\Mean{\bm{x}}$}}
     &     & \\    
\ar@{}[urr]|(.7){\hspace{.9ex}  \raisebox{-.9ex}{\txt{\scriptsize
$\Comp \rhd$}}}
\ar@{}[urr]|(.9){ \raisebox{-4ex}{\txt{\scriptsize 
$\vdots$}}}
        &      &   \\
\ar@{}[urr]|(.7){\hspace{.9ex}  \raisebox{-.9ex}{\txt{\scriptsize
$\Comp \rhd$}}}
     &   &     
\save"1,1"."3,2"*[F]\frm{}\restore  
}
\end{minipage}
\hspace{2ex} resp. \hspace{2ex}
\begin{minipage}{0.2\hsize}
\xymatrix@C=1pc@R=.5pc
{ 
&  \ar@{}[ddr] |{\txt{\normalsize $\Mean{\bm{x}}$}}
      &  
  \\    
 &              &  
\ar@{}[ull]|(.8){\hspace{.9ex}  \raisebox{-.9ex}{\txt{\scriptsize $\lhd
 \Comp$}}}
\ar@{}[ull]|(.9){ \raisebox{-4ex}{\txt{\scriptsize 
$\vdots$}}}
     \\
     &  &     
\ar@{}[ull]|(.8){\hspace{.9ex}  \raisebox{-.9ex}{\txt{\scriptsize 
$\lhd \Comp$}}}\save"1,2"."3,3"*[F]\frm{}\restore  
}   
\end{minipage}
}
}
\opencutleft
\begin{cutout}{0}{0pt}{\dimexpr.7\linewidth\relax}{3}
\vspace{1ex}
In the sequel, the two abbreviations are pictured as in the left
hand respectively.
Using a notation of the $r$-ary tensor folding of $\rhd$
(resp. $\lhd$), it is also written by
$\rhd^r \Comp \Mean{\bm{x}}$
(resp. $\Mean{\bm{x}} \Comp \lhd^r$).
This convention is equally employed
when indicated occurrences of assoc-(co)rets
are clear from the context.
\end{cutout}

Under this convention,
for any contracted component in the co-domain (resp. domain) of
$\Mean{\bm{x}}$, it holds;
\begin{align}
\lhd \Comp \rhd \Comp \Mean{\bm{x}}=\Mean{\bm{x}}
&  & 
(resp. \quad \Mean{\bm{x}} \Comp \lhd \Comp \rhd = \Mean{\bm{x}})
\label{asrcr}
\end{align}
That is, 
$\lhd \Comp \rhd$
composes (resp. precomposes) with any retraction
(resp. co-retraction) as the identity.
In other word, $\lhd \Comp \rhd$ is a projector on a contracted component
in the codomain (resp. domain) by composition (resp. precomposition).
Pictorially,    
\begin{figure}[!htbp]
\begin{align*}
\begin{minipage}[c]{0.2\hsize}
\xymatrix@C=1pc@R=.5pc
{      
\ar@{}[ddr] |(.6){\txt{\normalsize $\Mean{\bm{x}}$}}
\ar@{}[drr] | {\hspace{.9ex}  \raisebox{-.9ex}{\txt{\scriptsize $\labelcoret{i}$}}}    
  &   & \\    
        &   \ar@{}[r]^(.8){ \,\, \, \Comp \rhd \Comp \lhd}   &   \\
     &  &    
\save"1,1"."3,2"*[F]\frm{}\restore 
}
\end{minipage}
\hspace{0.02\textwidth}
\begin{minipage}[c]{0.05\hsize}
=
\end{minipage}
 & 
\begin{minipage}[c]{0.2\hsize}
\xymatrix@C=1pc@R=.5pc
{     
\ar@{}[ddr] |(.6){\txt{\normalsize $\Mean{\bm{x}}$}}
\ar@{}[drr] | {\hspace{.9ex}\txt{\scriptsize $\labelcoret{i}$}}    
     &   & \\    
        &      &   \\
     &  &    
\save"1,1"."3,2"*[F]\frm{}\restore 
}
\end{minipage}
& resp. 
\begin{minipage}[c]{0.2\hsize}
\xymatrix@C=1pc@R=.5pc
{   
& \ar@{}[ddr] |(.6){\txt{\normalsize $\Mean{\bm{x}}$}}
      &    \ar@{}[dll] | {\hspace{.1ex}  
\raisebox{-.9ex}{\txt{\scriptsize $\labelret{i}$}}}     \\    
 \ar@{}[r]^(.2){\rhd \Comp \lhd \Comp }
 &              &   \\
     &  &     
\save"1,2"."3,3"*[F]\frm{}\restore 
}
\end{minipage}
\hspace{0.05\textwidth}
\begin{minipage}[c]{0.05\hsize}
=
\end{minipage}
\hspace{-4ex}
\begin{minipage}[c]{0.1\hsize}
\xymatrix@C=1pc@R=.5pc
{     
& \ar@{}[ddr] |(.6){\txt{\normalsize $\Mean{\bm{x}}$}}
     &  
\ar@{}[dll] | {\hspace{.1ex}  
\raisebox{-.9ex}{\txt{\scriptsize $\labelret{i}$}}}
    \\    
  &      &         \\
  &   &      
\save"1,2"."3,3"*[F]\frm{}\restore 
}
\end{minipage}
\end{align*}
\caption{Equation (\ref{asrcr})}
\end{figure}
In Equation~(\ref{asrcr}),
the assoc-coret (resp. assoc-ret) occurs explicitly
as the last composed $\lhd$ (resp. the first precomposed $\rhd$).
Thus, the left most $\lhd$ (resp. right most $\rhd$)
in (\ref{asrcr}) is seen labelled
$\labelcoret{i}$ (resp. $\labelret{i}$)
such that the tag of $x_i$ is a formula occurrence in $\Gamma$,
where $\bm{x}=(x_1, \ldots ,x_\ell)
 \in \abs{\pf{\pi}{\Delta}{\Gamma}}$.
Since every contracted component occurs
pairwise 
in the co-domain and the domain of $\Mean{\bm{x}}$,
the two equations can be
written successively all at once;
\begin{align*}
\labelcoret{i} \Comp \rhd \Comp \Mean{\bm{x}} \Comp \lhd \Comp \labelret{i} =
 \Mean{\bm{x}} 
\end{align*}
\noindent 
The co-domain and the domain $U^\ell$ of $\Mean{\bm{x}}$
have in general several contracted components $U$s
labelled with $x_i$s, where $i$ ranges in
the set $\mathfrak{r}$ of the contraced $i$s.
Thus,
for all the several pairs of contracted components
in the domain and the co-domain of $\Mean{\bm{x}}$,
the parallel compositions and precompositions with
$\lhd \Comp
\rhd$s to each contracted components 
act as the identity on $\Mean{\bm{x}}$:
\begin{align}
 \Mean{\bm{x}} &  = 
(\otimes_{i \in  \mathfrak{r}} (\labelcoret{i} \! \Comp \rhd) )
\Comp \Mean{\bm{x}} \Comp
( \otimes_{i \in  \mathfrak{r}} (\lhd \Comp \labelret{i})) \nonumber \\  
 & = 
(\lhd \Comp \rhd)^r \Comp \Mean{\bm{x}} \Comp (\lhd \Comp \rhd)^r =
(\lhd^r \! \Comp \rhd^r) \Comp \Mean{\bm{x}} \Comp (\lhd^r \! \Comp \rhd^r)
\label{pluasrcr}
\end{align}
where $r$ is the cardinality of $\mathfrak{r}$,
hence is a number of the contracted components,
and the third equality is by $(\lhd \Comp \rhd)^r=\lhd^r \! \Comp \rhd^r$,
as $(-)^r$ is the $r$-ary tensor folding.
Since 
the last composed $\lhd^r$ (resp. the first precomposed $\rhd^r$)
in the rightmost expression of (\ref{pluasrcr})
are the explicit occurrences
of the assoc-rets (resp. assoc-corets),
the endomorphism $\Mean{\bm{x}}$ 
is written so that
all the assoc-rets $\rhd^r$ and the assoc-corets $\lhd^r$
can be made explicit as follows:
\begin{align} 
\Mean{\bm{x}}= 
\lhd^r \Comp \Mean{\bm{x}}^{\mathrm{o}}
\Comp \rhd^r  & & \txt{where} && 
\Mean{\bm{x}}^{\mathrm{o}} := 
\rhd^r \Comp \Mean{\bm{x}}
\Comp \lhd^r \label{mutual}
\end{align}
Roughly speaking, $\Mean{\bm{x}}^{\mathrm{o}}$ is $\Mean{\bm{x}}$
stripped of all the assoc-rets and assoc-corets. which is depicted in the
following Figure \ref{figmutual}:

\begin{figure}[!hbtp]
\begin{align*}
\scalebox{.8}{
\begin{minipage}[c]{0.2\hsize}
\xymatrix@C=.3pc@R=.1pc{
\ar@{}[ddddrr]|(.8){\txt{\normalsize $\Mean{\bm{x}}$}}
\ar@{}[dddrrr]|(.95){\vdots}
& & & 
\ar@{}[dddlll]|(.95){\vdots}
\\
& & &  \\ 
\rhd
 &  
& &  \lhd  \\
  &  &  &      \\
 \rhd &   &  &   \lhd \\
&   &  \\
 & & &   
\save"1,1"."7,4"*[F]\frm{}\restore
  }
\end{minipage}
}
\hspace{0.04\textwidth}
\begin{minipage}[c]{0.08\hsize}
=
\end{minipage}
\scalebox{.85}{
\begin{minipage}[c]{0.2\hsize}
\xymatrix@C=.3pc@R=.1pc{
& *=0{} 
\ar@{-}[rrr]
\ar@{-}[dd]
\ar@{}[ddddrr]|(.8){\txt{\normalsize
$\Mean{\bm{x}}^{\mathrm{o}}$}}
& & & 
 *=0{}  \ar@{-}[dd] & 
\\
& & & & &    \\ 
& \rhd
\ar@{-}[dd]
 &  
 & &  \lhd \ar@{-}[dd] &   \\
*=0{\vdots} &    &   &  &     & *=0{\vdots}   \\
&  \rhd &   &  &   \lhd &  \\
&  &      &   &  & \\
& *=0{} 
\ar@{-}[uu]
\ar@{-}[rrr] & & &  *=0{}  \ar@{-}[uu]  & 
  }
\end{minipage}
}
\hspace{0.02\textwidth}
\begin{minipage}[c]{0.08\hsize}
where
\end{minipage}
\scalebox{.85}{
\begin{minipage}[c]{0.2\hsize}
 \xymatrix@C=.3pc@R=.1pc{
& 
 *=0{} 
\ar@{-}[rrr]
\ar@{-}[dd]
\ar@{}[ddddrr]|(.8){\txt{\normalsize
$\Mean{\bm{x}}^{\mathrm{o}}$}}
 & & &  
 *=0{}  \ar@{-}[dd] 
& \\
& &  & & &    \\ 
& \mbox{\phantom{$\rhd$}}  \ar@{-}[dd]
  &  
 & &     \ar@{-}[dd] \mbox{\phantom{$\lhd$}} &  \\
*=0{\hspace{3ex}\vdots} &    &   &  &     & *=0{\hspace{-2ex}\vdots}    \\
& \mbox{\phantom{$\rhd$}}  &   &  &  \mbox{\phantom{$\lhd$}}  \\   &
& &  &   &   \\ 
& 
 *=0{} 
\ar@{-}[uu]
\ar@{-}[rrr] & & &  *=0{}  \ar@{-}[uu] &  
  }
\end{minipage}
}
\begin{minipage}[c]{0.05\hsize}
=
\end{minipage}
\scalebox{.8}{
\begin{minipage}[c]{0.2\hsize}
 \xymatrix@C=.3pc@R=.1pc{
&  \ar@{}[ddddrr]|(.8){\txt{\normalsize $\Mean{\bm{x}}$}}
& & & 
&  \\
&   & & &  &  \\ 
*=0{\lhd \Comp} &    \rhd
 &  
& &  \lhd  & *=0{\Comp \rhd}   \\
*=0{\vdots}
 &   &  &  &      & *=0{\vdots}  \\
*=0{\lhd \Comp}  &   \rhd &   &  &   \lhd & *=0{\Comp \rhd}  \\
 &  &  &   & &   \\
 & &     & &    
& \save"1,2"."7,5"*[F]\frm{}\restore 
  }
\end{minipage}
}
\end{align*}
\caption{Equation (\ref{mutual})}
\label{figmutual}
\end{figure}
\subsubsection{The action $\zr{\bm{x}}$ annihilating  associated
   (co)retractions} \label{subsub2}~\\
This subsection is concerned with defining the 
action $\zr{\bm{x}}$ (Definition \ref{zeroret})
over the associated retractions (resp. co-retractions)
in Definition \ref{endopartia} above.
The action $\zr{\bm{x}}$
arises from $\sigma_{\bm{x}}$
of (\ref{sigma}) when the feedback on the trace of ${\cal C}$
is taken into account,
and annihilates, using the zero morphism $0_U$,
a certain class of retractions and co-retractions.
This class is defined in Definition \ref{zeroret} 
in terms of zero input and output.



In what follows, we shall see
how feedback stemming from Gentzen cut-elimination
for a $\MALL$ proof $\pi$
acts on the assoc-rets
and the assoc-corets of $\Mean{\bm{x}}$ for
$\bm{x} \in \abs{\pf{\pi}{\Delta}{\Gamma}}$.
The action is stipulated in terms of the zero morphism
added in our framework. First,
in a categorical framework of Girard's GoI project,
the feedback is modelled by
the trace structure (cf. \cite{HS06})
defined by the seven axioms below:
\begin{align}
{\sf Tr}_{X,Y}^Z : 
{\cal C}(X \otimes Z, Y \otimes Z) \longrightarrow 
{\cal C}(X, Y) \label{trfam}
\end{align} 
There are three kinds of naturality axioms: {\em naturality} in $X$ and
{\em naturality} in $Y$,
and {\em dinaturality} in $Z$. The other axioms are {\em vanishing I,II},
{\em superposing} and {\em yanking}. See
Appendix \ref{traceaxioms} for the seven axioms:
the three naturalities and the four axioms.

In our setting of Definition \ref{endopartia},
the endomorphism 
$\Mean{\bm{x}}$ is on $U^{n+2m}$ so that $n$ and $2m$
are the numbers of formulas respectively in $\Gamma$ and in
a sublist $\widehat{\Delta}$, and $\sigma_{\bm{x}}$
is on the subfactor $U^{2m}$. Then the feedback is calculated by;
\begin{align}
\ex{\sigma}{\bm{x}} :=
\TR{(\operatorname{Id} \otimes \sigma_{\bm{x}})
 \Comp \Mean{\bm{x}}}{U^{2m}}{U^{n}}{U^{n}} \label{oriex}
\end{align}
See Figure \ref{firstfigs2}.
\begin{figure}[!hbt]
$$\xymatrix
@W=1pc @H=.1pc @R=.1pc
{  &  \hspace{-3ex}   \rhd \ar@<1.5ex>[r]
& 
\hspace{3.2ex}  \lhd \\
 \Gamma \hspace{-10ex} \vspace{5ex}  &  &  & \hspace{-12ex} \Gamma &      \\
 &  \hspace{-3ex} \rhd  &   \hspace{3.2ex} 
\lhd
 \\
  &  &  & &   
\ar`r[ddd]+/r3pc/`[ddd]+/d2pc/ `[ulll]+/l3pc/`[lll][lll]
  &  \\ 
\widehat{\Delta} \hspace{-10ex} \vspace{5ex}  & & \ar@<2.5ex>[r]_{\txt{\normalsize $\widehat{\Delta}$}} \ar@<-2.5ex>[r]  & 
   &
 &  \hspace{-10ex} \widehat{\Delta} 
\save"4,4"."6,5"*[F]\frm{}\restore    
 & \\
& \ar@<-1.3ex>[r]   &  \ar@{}[uuuuul]|(.5){\txt{\normalsize
 $\Mean{\bm{x}}$}}
& &  \ar@{}[uul]|{\txt{\normalsize$\sigma_{\bm{x}}$}}
\save"1,2"."6,3"*[F]\frm{}\restore 
\ar`r[d]+/r2pc/`[d]+/d1pc/ `[ulll]+/l2pc/`[lll][lll]
 &  \\
 & & & & &  }
$$
\caption{Equation(\ref{oriex}):
$\ex{\sigma}{\bm{x}}$ with feedback}
\label{firstfigs2}
\end{figure}

\smallskip

Note that when $\bm{x} \in \abs{\pf{\pi}{\Delta}{\Gamma}}$
comes from a proof $\pi$ of the multiplicative fragment,
the equation is exactly the GoI interpretation of the
proof $\pf{\pi}{\Delta}{\Gamma}$ (cf. \cite{HS06}).
This is because in the multiplicative fragment,
the index set $I$ becomes redundantly the singleton $\{ * \}$,
thus $\abs{ \pf{\pi}{\Delta}{\Gamma}}=\{ \bm{x}  \}$,
whereby $\sigma_{\bm{x}}$ is a simple tensor folding of
the symmetry $s$ (free of $0$ morphism).

By the naturalities of traces,
the assoc-corets (resp. the assoc-rets) of
$\Mean{\bm{x}}$
commute with
${\sf Tr}^{U^{2m}}_{U^{n}, U^{n}}$, hence 
taking a trace of
(\ref{asrcr}) composed with
$\operatorname{Id} \otimes \sigma_{\bm{x}}$
yields for any $i$ such that $i$-th component
of $\bm{x}=(x_1, \ldots ,x_\ell)$ is contracted.
\begin{align*}
\TR{
(\operatorname{Id} \otimes \sigma_{\bm{x}})
 \Comp (
\labelcoret{i} \Comp \rhd \Comp \Mean{\bm{x}}
)}{U^{2m}}{U^{n}}{U^{n}}
& =  
\labelcoret{i} \Comp
\TR{
(\operatorname{Id} \otimes \sigma_{\bm{x}})
 \Comp (\rhd \Comp \Mean{\bm{x}})}{U^{2m}}{U^{n}}{U^{n}}  \\ \nonumber
& = 
\labelcoret{i} \Comp \rhd \Comp \ex{\sigma}{\bm{x}}
\nonumber  
\\
 (\mbox{resp.} \, 
\TR{
(\operatorname{Id} \otimes \sigma_{\bm{x}})
 \Comp (
\Mean{\bm{x}} \Comp \lhd \Comp \labelret{i}
)}{U^{2m}}{U^{n}}{U^{n}} 
& = 
\ex{\sigma}{\bm{x}} \Comp \lhd \Comp \labelret{i}). \nonumber
\end{align*}
Thus all the assoc-rets and assoc-corets of $\Mean{\bm{x}}$ are written
explicitly;
\begin{align}
\TR{
(\operatorname{Id} \otimes \sigma_{\bm{x}})
 \Comp 
\Mean{\bm{x}}
}{U^{2m}}{U^{n}}{U^{n}}
& = 
(\otimes_{i \in \mathfrak{r}}  \, \labelcoret{i})  \Comp
\TR{ (\operatorname{Id} \otimes \sigma_{\bm{x}})
 \Comp \Mean{\bm{x}}^{\mathrm{o}} }{U^{2m}}{U^{n}}{U^{n}}
\Comp 
(\otimes_{i \in  \mathfrak{r}} \,  \labelret{i} ) \nonumber \\
& =
\lhd^r \Comp
\TR{ (\operatorname{Id} \otimes \sigma_{\bm{x}})
 \Comp \Mean{\bm{x}}^{\mathrm{o}} }{U^{2m}}{U^{n}}{U^{n}}
\Comp \rhd^r \label{allnatasrcr} 
\end{align}
where 
$\mathfrak{r}$ is the set of all 
contracted $i$s and $r$ is the cardinality of $\mathfrak{r}$. \\
Equation(\ref{allnatasrcr}) is depicted in Figure 
\ref{figallnatasrcr},
in which the dotted squares are the scopes of the traces
and the shifting of the scopes are naturalities of the $\rhd$s and the $\lhd$s;
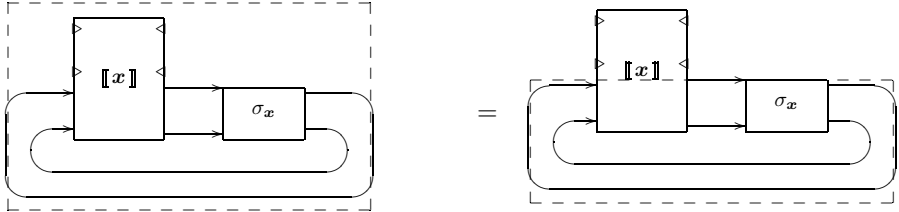
\begin{figure}[!htbp]
\begin{minipage}{0.35\hsize}
\scalebox{.8}{
$\xymatrix
@W=.1pc @H=.1pc @R=.1pc
{  & & & & &  \\
 &  \hspace{-1ex}   \triangleright 
& 
\hspace{1ex}  \triangleleft \\
  \hspace{-10ex} \vspace{5ex}  &  &  & \hspace{-12ex}  &      \\
 &  \hspace{-1ex} \triangleright  &   \hspace{1ex} 
\triangleleft 
 \\
  &  &  & &   
\ar`r[ddd]+/r3pc/`[ddd]+/d2pc/ `[ulll]+/l3pc/`[lll][lll]
  &  \\ 
 \hspace{-10ex} \vspace{5ex}  & & \ar@<2.5ex>[r] \ar@<-2.5ex>[r]  & 
   &
 &  \hspace{-10ex} 
\save"5,4"."7,5"*[F]\frm{}\restore    
 & \\
& 
&  \ar@{}[uuuuul]|(.5){\txt{\normalsize
 $\Mean{\bm{x}}$}}
& &  \ar@{}[uul]|{\txt{\normalsize$\sigma_{\bm{x}}$}}
\save"2,2"."7,3"*[F]\frm{}\restore 
\ar`r[d]+/r2pc/`[d]+/d1pc/ `[ulll]+/l2pc/`[lll][lll]
 &  \\
 & & & & &  \\  & & & & &  \\  & & & & &  \\
& & & & &  \\
\save"1,1"."11,6"*[F--]\frm{}\restore    
  }
$}
\end{minipage} 
\begin{minipage} {0.1\hsize}
\hspace{6ex} =
\end{minipage} 
\begin{minipage}{0.4\hsize}
\scalebox{.8}{
$
\xymatrix
@W=.1pc @H=.1pc @R=.1pc
{ &  \hspace{-1ex}   \triangleright 
&   
\hspace{1ex}  \triangleleft \\
  \hspace{-10ex} \vspace{5ex}  &  &  & \hspace{-12ex}  &      \\
 &  \hspace{-1ex} \triangleright  &   \hspace{1ex} 
\triangleleft 
 \\
  &  &  & &   
\ar`r[ddd]+/r3pc/`[ddd]+/d2pc/ `[ulll]+/l3pc/`[lll][lll]
  &  \\ 
 \hspace{-10ex} \vspace{5ex}  & & \ar@<2.5ex>[r] \ar@<-2.5ex>[r]  & 
   &
 &  \hspace{-10ex} 
\save"4,4"."6,5"*[F]\frm{}\restore    
 & \\
& 
&  \ar@{}[uuuuul]|(.5){\txt{\normalsize
 $\Mean{\bm{x}}$}}
& &  \ar@{}[uul]|{\txt{\normalsize$\sigma_{\bm{x}}$}}
\save"1,2"."6,3"*[F]\frm{}\restore 
\ar`r[d]+/r2pc/`[d]+/d1pc/ `[ulll]+/l2pc/`[lll][lll]
 &  \\
 & & & & &  \\
 & & & & &  \\  & & & & &  \\  & & & & &  \\
\save"4,1"."10,6"*[F--]\frm{}\restore    
 }
$}
\end{minipage}
\caption{Equation (\ref{allnatasrcr}) : naturality of assoc-rets and assoc-corets} 
\label{figallnatasrcr}
\end{figure}

While inside the sole $\Mean{\bm{x}}$,
the assoc-rets and the assoc-corets (written explicitly in (\ref{mutual}))
do not interact with zero morphisms  because
the construction of $\Mean{\bm{x}}$ of Definition \ref{endopartia}
is free from the zero morphisms.
Remember that the zero morphisms reside only in $\sigma_{\bm{x}}$
as subfactors (cf.(\ref{sigma})).
However when they are put inside the context
$\TR{
(\operatorname{Id} \otimes \sigma_{\bm{x}})
 \Comp -}{U^{2m}}{U^{n}}{U^{n}}$
(written explicitly in (\ref{allnatasrcr})),
they may interact with zero morphisms arising
from $\sigma_{\bm{x}}$ 
via the feedback of the trace.
That is, the trace in a monoidal category
takes feedback into account, hence
makes the zeros stemming from $\sigma_{\bm{x}}$
interact with 
the assoc-rets and
the assoc-corets of $\Mean{\bm{x}}$. 
This yields a certain action $\epsilon_{\bm{x}}$ on
the assoc-(co)rets of $\Mean{\bm{x}}$, 
as defined in Definition \ref{zeroret} below.



\begin{definition} \label{zio}
{\bf zero input (resp. output) of
assoc-coret (resp. assoc-ret) w.r.t the interpretation
$\sigma_{\bm{x}}$ of cuts} 

\noindent ({\bf zero input of assoc-coret $\lhd$}) ~
An assoc-coret $\lhd$ of $\Mean{\bm{x}}$
is said to have {\em zero input} w.r.t $\sigma_{\bm{x}}$
when $\lhd$ decomposes in $\ex{\sigma}{\bm{x}}$ either as
$\lhd \Comp (0_U \! \otimes U)$ or as
$\lhd \Comp (U \otimes 0_U)$.

\noindent ({\bf zero output of assoc-ret $\rhd$}) ~
An assoc-ret $\rhd$ of $\Mean{\bm{x}}$
is said to have {\em zero output} w.r.t $\sigma_{\bm{x}}$
when $\rhd$ decomposes in $\ex{\sigma}{\bm{x}}$ either
as $(0_U \! \otimes U) \Comp \rhd$ or as $(U \otimes 0_U) \Comp \rhd$.
\end{definition}
Why do we use the terminology zero input (resp. output) ? 
The object $U \otimes U$ of $\lhd$'s domain (resp. $\rhd$'s co-domain)
can be regarded as having two inputs (resp. outputs),
one left component $U$ and the other right one.
Then the decomposition in each case says that one of two inputs (resp. outputs)
is zero. \\
Pictorially,
\scalebox{.8}{
$\begin{array}{r}
\stackrel{0_U}{\longrightarrow}   \\
 \stackrel{\text{\normalsize $\longrightarrow$}}{\text{\scriptsize $U$}} 
  \end{array}
\! \! \! \raisebox{.4ex}{$\lhd$}$
} or
\scalebox{.8}{
$ \begin{array}{r}
\stackrel{U}{\longrightarrow}   \\
 \stackrel{\text{\normalsize $\longrightarrow$}}{\text{\scriptsize $0_U$}} 
  \end{array}
\! \! \! \raisebox{.4ex}{$\lhd$}$}
for the zero input
and 
\scalebox{.8}{
$\raisebox{.4ex}{$\rhd$} \! \! \!
\begin{array}{l}
\stackrel{0_U}{\longrightarrow}   \\
 \stackrel{\text{\normalsize $\longrightarrow$}}{\text{\scriptsize $U$}} 
\end{array}$
} or
\scalebox{.8}{
$\raisebox{.4ex}{$\rhd$} \! \! \!
\begin{array}{l}
\stackrel{U}{\longrightarrow}   \\
 \stackrel{\text{\normalsize $\longrightarrow$}}{\text{\scriptsize $0_U$}} 
\end{array}$
}
for the zero output.


Note that Definition \ref{zio} is alternatively stated as follows:
When the assoc-coret (resp. assoc-ret)
is written explicitly as $\ex{\sigma}{\bm{x}} = \lhd \Comp {\sf g}$
(resp. $\ex{\sigma}{\bm{x}} = {\sf g} \Comp \rhd$) 
(cf. (\ref{asrcr})),
the assoc-coret $\lhd$ (resp. assoc-ret $\rhd$)
has zero input (resp. output) iff
either $0_U \otimes U$
or $U \otimes 0_U$ acts trivially 
 on ${\sf g}$ by composing
(resp. precomposing) to the
indicated component $U \otimes U$.


\begin{example} \label{exzio}
Let $\pi_{[A \& A, A^\perp \oplus A^\perp] \, \,
A^\perp, A \otimes B, B^\perp}
$ be a proof obtained by
a $\otimes$-rule
between $\pi_1$ of Section \ref{prolo} (the first paragraph)
and $ax_{B^\perp, B}$.
Let $\bm{x} := \nu_2 \times (\bar{\star}, \star)
\in \abs{\pi}$, where
$\nu_2$ is in Section \ref{prolo} and $(\bar{\star}, \star) \in \abs{
ax_{B^\perp, B}}$. 
\bigskip
\renewcommand\windowpagestuff{%
\flushleft
\vspace{3ex}
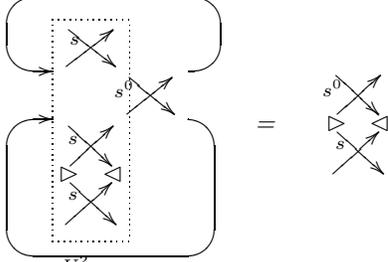
\begin{minipage}{0.4\hsize}
\xymatrix@C=1.3pc@R=1pc
{&   \ar[rd] &       &    &   \\
& \ar[ru]^(.4){s} \ar@{<-} `l []-<2em,0cm> `[u]+<0cm,1em> `[rr]+<1.5em,0cm> `[l] [rr] 
 &       \ar[rd]  &    &          \\
&    \ar[rd]   & \ar[ru]^(.3){s^0}   &
\ar `r []-<1.5em,0cm> `[dd]-<0cm,1em> `[ll]-<2em,0cm>   `[r] [ll] 
      &   \\
&  \rhd \hspace{-1.5ex}  \ar[ru]^(.4){s}  
\ar[rd]  &  \hspace{-1.5ex} \lhd  & &    \\
 & \ar[ru]^(.4){s}  &     &   & 
\save"1,2"."5,3"*[F.]\frm{}\restore }
\end{minipage}
\begin{minipage}{0.1\hsize}
=
\end{minipage}
\begin{minipage}{0.1\hsize}
\xymatrix@C=1.3pc@R=1pc
{   \ar[rd]   &  \\
 \rhd \hspace{-1.5ex}  \ar[ru]^(.35){s^0}  
\ar[rd]  &  \hspace{-1.5ex} \lhd \\
 \ar[ru]^(.4){s}  & 
   & 
   }
\end{minipage}
}\opencutleft
\begin{cutout}{0}{0pt}{\dimexpr.55\linewidth\relax}{7}
\noindent
\bigskip
Then $\Mean{\bm{x}}$ has the unique pair of assoc-ret and
assoc-coret both interpreting the $\otimes$-rule.
See the left-hand dotted rectangle representing
$\Mean{x}$ with the assoc-ret and the assoc-coret.
The pair of assoc-ret and the assoc-coret appears explicitly
in the second and the third
of the following equations (in which 
$\sigma_{\bm{x}}=s^0_{U,U}$):
\end{cutout}
\bigskip
$\begin{aligned}
\TR{
(\operatorname{Id} \otimes \sigma_{\bm{x}})
 \Comp \Mean{\bm{x}}}{U^2}{U^3}{U^3}
= 
\lhd \Comp
\TR{ (\operatorname{Id} \otimes \sigma_{\bm{x}})
 \Comp \Mean{\bm{x}}^{\mathrm{o}}}{U^2}{U^4}{U^4}
\Comp \rhd
= 
\lhd \Comp
(s_{U,U}^0 \otimes s_{U,U})
\Comp \rhd
\end{aligned}$, 
where 
$\Mean{\bm{x}}^{\mathrm{o}}$ is 
$\Mean{\bm{x}}$ without
the assoc-ret and the assoc-coret. \\
See the above figure whose LHS and RHS are 
the first and the last equations, respectively,
The assoc-coret (resp. assoc-ret) has zero input (resp. output)
because the right picture depicts
$\lhd$ (resp. $\rhd$) having a zero input (resp. output) 
from the northwest (resp. to the northeast).
Hence $(0_U \otimes U)$ composes (resp. precomposes)
to $s_{U,U}^0 \otimes s_{U,U}$ trivially. 
\end{example}

\begin{definition}[action $\zr{\bm{x}}$ on
assoc-rets and assoc-corets of $\Mean{\bm{x}}$] \label{zeroret} 
The endomorphism
$\sigma_{\bm{x}}$ of Definition \ref{endopartia} for $\bm{x} \in
 \abs{\pf{\pi}{\Delta}{\Gamma}}$ yields the  following action
$\zr{\bm{x}}$ on the assoc-rets and the assoc-corets
of $\Mean{\bm{x}}$.
The action $\epsilon_{\bm{x}}$ acts
on each assoc-ret and assoc-coret as {\em either zero}
 or the {\em identity} by (pre)composition on
them, as follows: 
\begin{center}
\begin{tabular}{lcr}
 
$ \rhd^{\zr{\bm{x}}} =
\begin{cases}
 \rhd^0 &  \text{if $\rhd$
 has a zero output} \\
& \hfill \text{w.r.t $\sigma_{\bm{x}}$}  \\
 \rhd & \text{otherwise}
\end{cases}$
&  \hspace{3ex} &
$\lhd^{\zr{\bm{x}}}   =
\begin{cases}
 \lhd^0 &  \text{if $\lhd$
has a zero input} \\
 &  \hfill \text{w.r.t $\sigma_{\bm{x}}$}   \\
\lhd  & \text{otherwise}
\end{cases}$
\end{tabular}
\end{center}
where {\em zero actions $\rhd^0$ and $\lhd^0$} are
defined respectively as follows:
\begin{align*}
\rhd^0 :=
0_{U, U\otimes U}
= k \, \Comp \, 0_U 
= (0_{U} \otimes 0_{U}) \, \Comp \, k
& & 
\lhd^0 :=
0_{U\otimes U, U} = 0_U \, \Comp \, j = j \, \Comp \, (0_{U}
	       \otimes 0_{U})  
\end{align*}
That is, the zero
{\em annihilates}
 the pair of assoc-ret and assoc-coret 
$j: U \otimes U \lhd U :  k$
to the pair of the zero morphisms $j^0: U \otimes U \lhd^0 U: k^0$,
where $j^0=\lhd^0$ and $k^0=\rhd^0$.



\end{definition}

The action $\epsilon_{\bm{x}}$ of Definition \ref{zeroret} is
by definition conjugate
on the pairwise tensor foldings
$(\otimes_{i \in \mathfrak{r}} \, \labelret{i}, \otimes_{i \in \mathfrak{r}} \, \labelcoret{i} )
=(\rhd^r, \lhd^r)$ 
of the assoc-rets and the assoc-corets represented in (\ref{mutual}),
where
$\mathfrak{r}= \{ i \mid 
\mbox{$i$ is contracted with 
$\bm{x}=(x_1, \ldots ,x_\ell)$} \}$.
Hence
we define to formulate the action on $\Mean{\bm{x}}$
by conjugation:
\begin{align}
 \Mean{\bm{x}}^{\zr{\bm{x}}} :=
(\lhd^{\zr{\bm{x}}})^r
\Comp
\Mean{\bm{x}}^{\mathrm{o}}
\Comp
(\rhd^{\zr{\bm{x}}})^r \label{eex}
\end{align}
Pictorially, 
$$\begin{aligned}
\begin{minipage}[c]{0.1\hsize} 
$\Mean{\bm{x}}^{\zr{\bm{x}}}$
\end{minipage}
\begin{minipage}[c]{0.05\hsize} 
:=
\end{minipage}
\scalebox{.8}{
\begin{minipage}[c]{0.2\hsize}
\xymatrix@C=.3pc@R=.1pc{
& *=0{} 
\ar@{-}[rrr]
\ar@{-}[dd]
\ar@{}[ddddrr]|(.8){\txt{\normalsize
$\Mean{\bm{x}}^{\mathrm{o}}$}}
& & & 
 *=0{}  \ar@{-}[dd] & 
\\
& & & & &    \\ 
& \rhd^{\zr{\bm{x}}}
\ar@{-}[dd]
 &  
 & &  \quad \lhd^{\zr{\bm{x}}} \ar@{-}[dd] &   \\
*=0{\hspace{3ex} \vdots} &    &   &  &     & *=0{\hspace{-3ex} \vdots}   \\
&  \rhd^{\zr{\bm{x}}} &   &  &  \quad  \lhd^{\zr{\bm{x}}} &  \\
&  &      &   &  & \\
& *=0{} 
\ar@{-}[uu]
\ar@{-}[rrr] & & &  *=0{}  \ar@{-}[uu]  & 
  }
\end{minipage}
}
\end{aligned}$$
This action of $\epsilon_{\bm{x}}$,
by naturalities, 
extends to the action $\epsilon_{\bm{x}}$
on the corresponding retractions and co-retractions
in (\ref{allnatasrcr}):
\begin{align*}
(\TR{( \operatorname{Id}  
 \otimes \sigma_{\bm{x}})
 \Comp \Mean{\bm{x}}}{U^{2m}}{U^{n}}{U^{n}})^{\zr{\bm{x}}} 
& :=
(\lhd^{\zr{\bm{x}}})^r
\Comp \TR{( \operatorname{Id}  
 \otimes \sigma_{\bm{x}})
 \Comp \Mean{\bm{x}}^{\mathrm{o}}}{U^{2m}}{U^{n}}{U^{n}}
(\rhd^{\zr{\bm{x}}})^r  \\
& = 
\TR{
(\lhd^{\zr{\bm{x}}})^r
\Comp
( \operatorname{Id}  
 \otimes \sigma_{\bm{x}})
 \Comp 
\Mean{\bm{x}}^{\mathrm{o}}
\Comp
(\rhd^{\zr{\bm{x}}})^r
}{U^{2m}}{U^{n}}{U^{n}} 
& \mbox{by nats}\\
& = 
\TR{( \operatorname{Id}  
 \otimes \sigma_{\bm{x}})
 \Comp 
(\lhd^{\zr{\bm{x}}})^r
\Comp
\Mean{\bm{x}}^{\mathrm{o}}
\Comp
(\rhd^{\zr{\bm{x}}})^r
}{U^{2m}}{U^{n}}{U^{n}} 
& \mbox{by (\ref{allnatasrcr})} \\
& = 
\TR{(\operatorname{Id} \otimes \sigma_{\bm{x}} )
 \Comp \Mean{\bm{x}}^{\zr{\bm{x}}}}{U^{2m}}{U^{n}}{U^{n}}
& \mbox{by (\ref{eex})}
 \end{align*}


Note by (\ref{oriex}) that the LHS of the first equation is
$\ex{\sigma}{\bm{x}}^{\zr{\bm{x}}}$. 
Recall that $r$ is the number of the assoc-rets $\labelret{i}$
(equally the assoc-corets $\labelcoret{i}$) of $\Mean{\bm{x}}$
such that $i$ is contracted.

\subsubsection{The Execution formula}
\begin{definition}[Execution formula $\Ex{\sigma}{\bm{x}}$
for $\bm{x} \in \abs{\pf{\pi}{\Delta}{\Gamma}}$] \label{Exx}~~\\
For every $\bm{x}
\in \abs{\pf{\pi}{\Delta}{\Gamma}}$,
the endomorphism $\Ex{\sigma}{\bm{x}}$ is defined by
\begin{align*}
\Ex{\sigma}{\bm{x}}
& := \ex{\sigma}{\bm{x}}^{\zr{\bm{x}}} \\
& =
\TR{
(\operatorname{Id} \otimes \sigma_{\bm{x}})
 \Comp \Mean{\bm{x}}^{\zr{\bm{x}}}}{U^{2m}}{U^{n}}{U^{n}},  
\end{align*}
where 
$(\Mean{\bm{x}}, \sigma_{\bm{x}})$
is the pair of the endomorphism on $U^{n+2m}$ and on the subfactor
$U^{2m}$ in Definition \ref{endopartia}
and $\zr{\bm{x}}$ is the action in Definition \ref{zeroret} 
on the assoc-rets and the assoc-corets of $\Mean{\bm{x}}$.
The domains (resp. the co-domains) of the assoc-rets (resp. the assoc-corets)
lie among the subfactor $U^n$ in the domain (resp. the co-domain) of $\Mean{\bm{x}}$.
See Figure \ref{extwofigs}.
\end{definition}

\begin{figure}[htbp]
\begin{center}
$\xymatrix
@W=1pc @H=.1pc @R=.1pc
{  &  \hspace{-3ex}   \rhd \ar@<1.5ex>[r]
& 
\hspace{3.2ex}  \lhd \\
 \Gamma \hspace{-10ex} \vspace{5ex}  &  &  & \hspace{-12ex} \Gamma &      \\
 &  \hspace{-3ex} \blacktriangleright  &   \hspace{3.2ex} 
\blacktriangleleft 
 \\
  &  &  & &   
\ar`r[ddd]+/r3pc/`[ddd]+/d2pc/ `[ulll]+/l3pc/`[lll][lll]
  &  \\ 
\widehat{\Delta} \hspace{-10ex} \vspace{5ex}  & & \ar@<2.5ex>[r]_{\txt{\normalsize $\widehat{\Delta}$}} \ar@<-2.5ex>[r]  & 
   &
 &  \hspace{-10ex} \widehat{\Delta} 
\save"4,4"."6,5"*[F]\frm{}\restore    
 & \\
& \ar@<-1.3ex>[r]   &  \ar@{}[uuuuul]|(.5){\txt{\normalsize
 $\Mean{\bm{x}}^{\zr{\bm{x}}}$}}
& &  \ar@{}[uul]|{\txt{\normalsize$\sigma_{\bm{x}}$}}
\save"1,2"."6,3"*[F]\frm{}\restore 
\ar`r[d]+/r2pc/`[d]+/d1pc/ `[ulll]+/l2pc/`[lll][lll]
 &  \\
 & & & & &  }
$
\end{center}
\caption{ Execution formula
$\Ex{\sigma}{\bm{x}}$,  where
   $\blacktriangleright$ (resp. $\blacktriangleleft$)
  denotes $\rhd^0$ (resp. $\lhd^0$)}
\label{extwofigs}
\end{figure}

\begin{example}
Let $\bm{x}$ be of
Example \ref{exzio}. 
Since $\zr{\bm{x}}$ acts as zero both on the unique assoc-ret $\rhd$
and on the unique assoc-coret $\lhd$,
$\begin{aligned}
\Ex{\sigma}{\bm{x}} =
\lhd^0 \Comp (s_{U,U}^0 \otimes s_{U,U})
\Comp \rhd^0= 0_{U^3,U^3}.
\end{aligned}$ 
\end{example}

Finally, the execution formula is run point-wise
for every enumerated set $\nu$ in interpretation of a proof in
$\Relc$.

\begin{definition}[Execution formula $\inEx{\nu}{J}$
for $\nu \in \abs{\pf{\pi}{\Delta}{\Gamma}}^J$]~~\\
Let $\pf{\pi}{\Delta}{\Gamma}$ be a $\MALLC$ proof.
For every  $\nu \in \abs{\pf{\pi}{\Delta}{\Gamma}}^J$,
$\Ex{\sigma}{\nu} \in \abs{\Gamma}^J$
is defined indexwise by:
\begin{align*}
(\, \inEx{\nu}{J} \, )_j =  
\ptEx{\nu_j} & \quad \text{for every index $j \in J$}
\end{align*}
\end{definition}

\subsection{Zero Convergence of Execution Formula}
This subsection concerns the main proposition
(Proposition \ref{convzero}), which says that
communicating two proofs via mismatched pair
yields zero convergence of Ex.
We start with the tracing zero lemma derivable
from some trace axioms. 

\begin{lemma}[tracing zero] \label{trz}
For any natural number $n \geq 1$, 
\begin{align} 
\TR{0_{U^{n+1}}}{U}{U^n}{U^n}&  =  0_{U^n} 
\end{align}
\end{lemma}
\begin{proof}{}
First by Lemma \ref{tenzero} $0_{U^m}=(0_U)^m$ for any natural number
 $m$.
Then by superposing $\TR{0_{U^{n+2}}}{U}{U^{n+1}}{U^{n+1}}=
0_{U^n}  \otimes \TR{0_{U^2}}{U}{U}{U}$, it suffices to prove
the assertion for $n=1$.
Second, observe
the equation\footnote{In a more general setting, the
natural iso
$(b \otimes U) \Comp s_{U,U} \Comp (a \otimes U)
\cong a \otimes b
\cong
(U \otimes a) \Comp s_{U,U} \Comp (U \otimes b)$,
for any endomorphisms $a$ and $b$ on $U$}
\begin{align}
(0_U \otimes U) \Comp s_{U,U} \Comp (0_U \otimes U)
= 0_U \otimes 0_U
= (U \otimes 0_U) \Comp s_{U,U} \Comp (U \otimes 0_U)
  \label{prps}
\end{align}
Thus $\TR{0_{U^2}}{U}{U}{U}=
0_U  \Comp \TR{s_{U,U}}{U}{U}{U} \Comp \,
0_U =
0_U \Comp U \Comp \, 0_U$,
where the first equation is by naturalities 
and the second equation is by yanking.
\end{proof}

\begin{figure}[!htb]
\begin{tabular}{lcr}
\begin{tabular}{c}
\begin{minipage}{0.3\hsize}
\xymatrix@C=1pc@R=.2pc
{   &     & & & &  \\
&       \hspace{-1ex}  \rhd
\ar@{}[dr] |(.6){\txt{\normalsize $\Excon_1$}}
     & \hspace{.5ex}  \lhd & & \\    
& \hspace{-1ex}  
\blacktriangleright
      & \hspace{.5ex}  \blacktriangleleft  &  & \\
& 
\ar@{<-} `l []-<2em,0cm> `[uuu]+<0cm,1em> `[rr]+<1.5em,0cm> `[l] [rr] 
&  \ar[rdd]^(.35){s^0}  &  
  &  
    &   \\
 & & & &  & = \bm{0}  \\
&  \ar@{}[ddr] |(.55){\txt{\normalsize $\Excon_2$}}  & \ar[ruu] & 
\ar `r []-<1.5em,0cm> `[ddd]-<0cm,1em> `[ll]-<2em,0cm>   `[r] [ll]   & \\
& \rhd
     &  \lhd \hspace{-1ex}
 & & \\    
 & \blacktriangleright      &    \blacktriangleleft \hspace{-1ex} 
 &    &  &  \\
  &     & & & & 
\save"2,2"."4,3"*[F]\frm{}\restore 
\save"6,2"."8,3"*[F]\frm{}\restore   
}
\end{minipage}
\\ \\
Proposition \ref{convzero}: $\Excon_i$ denotes
$\Ex{\sigma
}{\bm{x}_i}$. 
\end{tabular}
 & \hspace{1ex} & 
\begin{tabular}{c}
\begin{minipage}{0.3\hsize}
\xymatrix@C=1pc@R=.2pc
{  &  \hspace{-1ex}  \rhd
\ar@{}[dr] |{\txt{\normalsize ${\sf ex}^\delta$}}
     & \hspace{.5ex}  \lhd & \\    
& \hspace{-1ex}  \blacktriangleright      & \hspace{.5ex}  
\blacktriangleleft  & &  = \bm{0} \\
\ar[r]^0 & & 
      \ar[r]^0     &    
\save"1,2"."3,3"*[F]\frm{}\restore 
}
\end{minipage} \\ \\
Lemma \ref{sublem} : ${\sf ex}$ denotes
${\sf ex}(\sigma, \bm{x})$.
\\ \\
\begin{tabular}{c}
\begin{minipage}{0.3\hsize}
\xymatrix@C=1pc@R=.2pc
{   
\ar@{}[ddr] |{f} 
     &   &  & 
\\   
       & 
  &  & = 
\\
 &  
      \ar[r]^0     &   
\ar `r []-<1.5em,0cm> `[d]-<0cm,1em> `[ll]-<2em,0cm>   `[r] [ll]  & \\
 & & & 
\save"1,1"."3,2"*[F]\frm{}\restore 
}
\end{minipage} 
\begin{minipage}{0.3\hsize}
\xymatrix@C=1pc@R=.2pc
{  &  
\ar@{}[ddr] |{f} 
     &   & \\    
&       & 
  & &  
\\
\ar[r]^0 & & 
      \ar[r]^0     &    
\save"1,2"."3,3"*[F]\frm{}\restore \\
& & & \\ \\
}
\end{minipage} 
\\
$0$'s are $0_U$, $0_{I,U}$, $0_{U,I}$ from left to right. \\
Equation(\ref{vanishtr})
\end{tabular} 
 \end{tabular}
\end{tabular}
\caption{Prop \ref{convzero} and Lem \ref{sublem}
and Equation (\ref{vanishtr}) pictorially,
where $\blacktriangleright$ and $\blacktriangleleft$
denote respectively $\rhd^0$ and $\lhd^0$.}
\label{picconvzero}
\end{figure}

We prepare the following Lemma \ref{sublem}, which will directly entail
the main Proposition \ref{convzero}.

\begin{definition}[action $\delta_{\bm{x}}$] \label{actdelta}
For $\bm{x} \in \abs{\pf{\pi}{\Delta}{\Gamma}}$,
let us put $\Mean{\bm{x}}$ into the  context
$ (\operatorname{Id} \otimes \, 0_U ) \Comp  (-) \Comp
 (\operatorname{Id} \otimes \, 0_U )$,
allowing interaction of the assoc-rets and the assoc-corets of $\Mean{\bm{x}}$
with the two zeros $0_U$ in the context. 
Zero input (resp. zero output) of assoc-ret (resp. assoc-coret) in this context
is defined in the same manner, yielding
the action, say $\delta_{\bm{x}}$,
on the assoc-rets and the assoc-corets of $\Mean{\bm{x}}$
same as in Definition \ref{zeroret} (but simpler without the feed
 back): That is, the morphism $\rhd^{\delta_{\bm{x}}}$
is defined to be $\rhd^0$ (resp. $\rhd$)
if $\rhd$ decomposes in 
$(\operatorname{Id} \otimes \, 0_U ) \Comp \, 
\Mean{\bm{x}}
\Comp \,  (\operatorname{Id} \otimes \, 0_U )$
either as 
$(0_U \otimes \, U ) \Comp \rhd$
or 
$(U \otimes \, 0_U ) \Comp \rhd$
(resp. otherwise).
Symmetrically,
$\lhd^{\delta_{\bm{x}}}$
is defined to be $\lhd^0$ (resp. $\lhd$)
if $\lhd$ decomposes in 
$(\operatorname{Id} \otimes \, 0_U ) \Comp \, 
\Mean{\bm{x}}
\Comp \,  (\operatorname{Id} \otimes \, 0_U )$
either as 
$\lhd \Comp (0_U \otimes \, U )$
or 
$\lhd \Comp (U \otimes \, 0_U )$
(resp. otherwise).
\end{definition}

\begin{lemma}[lemma for Prop \ref{convzero}] \label{sublem}
\begin{align*}
 (\operatorname{Id} \otimes \, 0_U ) \Comp \, 
\ex{\sigma}{\bm{x}}^{\delta_{\bm{x}}}
\Comp \,  (\operatorname{Id} \otimes \, 0_U )
&  = 0_{U^n}, 
\end{align*}
where $\ex{\sigma}{\bm{x}}^{\delta_{\bm{x}}}$
is $\ex{\sigma}{\bm{x}}$ of (\ref{oriex})
whose $\Mean{\bm{x}}$ is replaced by 
$\Mean{\bm{x}}^{\delta_{\bm{x}}}$
using the action $\delta_{\bm{x}}$ of Definition \ref{actdelta}.

See Figure \ref{picconvzero} (upper-right) depicting the equation.
The lemma holds up to the the permutations $\tau$ on $U^n$
so that the left $\ex{\sigma}{\bm{x}}$ is read by   
$\tau^{-1} \Comp \, \ex{\sigma}{\bm{x}}
\Comp \, \tau$. Hence the assertion is independent of the choice of $U$ for
the $0_U$. The choice is of one formula occurrence from $\Gamma$,
as each occurrence is interpreted by the distinct $U$.
\end{lemma}
\begin{proof}{}
Induction on the construction of $\pi$ for $\bm{x}$
in Definition \ref{endopartia}.
In the proof, Equation(\ref{prps}) in the proof of Lemma \ref{trz} is used.
In the following, for $i=1,2$,  
$\bm{x}_i$ are the premises of $\bm{x}$ (i.e., 
$\bm{x}_1 = \bm{y}$ and
$\bm{x}_2 = \bm{z}$ 
in Definition
\ref{endopartia}), and 
$\excon_i$ denotes $\ex{\sigma}{\bm{x}_i}$.

\noindent (axiom) \\
$( \operatorname{Id} \otimes \, 0_U ) \Comp
\Mean{ax} \Comp 
( \operatorname{Id} \otimes \, 0_U ) =
( \operatorname{Id} \otimes \, 0_U ) \Comp
s_{U,U} \Comp 
( \operatorname{Id} \otimes \, 0_U )
= 0_U \otimes 0_U$.

\smallskip

\noindent ($\otimes$-rule) 
(case 1) $U$ is introduced by the $\otimes$-rule.\\
$\begin{aligned}
& (\operatorname{Id_1} \otimes \, 0_U \otimes \operatorname{Id_2}) 
\Comp \lhd \Comp ( \excon_1 \otimes \excon_2 ) \Comp \rhd
 \Comp
(\operatorname{Id_1} \otimes \, 0_U \otimes \operatorname{Id_2})  \\
& = 
\lhd \Comp ( ((\operatorname{Id_1} \otimes \, 0_U) \Comp \, {\sf ex_1}
\Comp (\operatorname{Id_1} \otimes \, 0_U) )
\otimes 
((\operatorname{Id_2} \otimes \, 0_U) \Comp \, {\sf ex_2}
\Comp (\operatorname{Id_2} \otimes \, 0_U) ) \Comp \rhd
= 0_{U^{n_1}} \otimes 0_{U^{n_2}}
\end{aligned}$ \\
The last equation is
by I.H.'s on $\Mean{\bm{x}_1}$ and $\Mean{\bm{x}_2}$.

\smallskip

\noindent ($\otimes$-rule) (case 2) other than case 1: \\
In this case, 
the $U$ for the $0_U$ of $\operatorname{Id}  \otimes 0_U$
is a factor from the (co)domain
of $\Mean{\bm{x}_i}$. 
We assume without loss of generality
that $i=1$. Then,
$(\operatorname{Id} \otimes \, 0_U)\Comp \, {\sf ex_1} \Comp
(\operatorname{Id} \otimes \, 0_U)=0_{U^{n_1}}$
by I.H on $\bm{x}_1$.
This directly implies that 
the co-retraction and the retraction $(j,k)$
interpreting the $\otimes$-rule are acted by $\delta$ as zero,
denoted by $(j^0, k^0)$, 
since $j$'s output and $k$'s input both on $\Mean{\bm{x}_1}$
are zeros by the I.H. Hence, when $(j,k)$ is written by $(\lhd,\rhd)$, \\
$\begin{aligned}
\lhd^0 \Comp (0_{U^{n_1}} \otimes {\sf ex_2})
\Comp \rhd^0
& =   
\lhd \Comp (0_{U^{n_1}} \otimes 
(0_U \otimes \operatorname{Id}) \Comp \, {\sf ex_2}
\Comp (0_U \otimes \operatorname{Id}))
\Comp \rhd \\
& =
\lhd \Comp (0_{U^{n_1}} \otimes 0_{U^{n_2}} ) \Comp \rhd
= 0_{U^n} 
\end{aligned}$ \\
The first equation is by the assumption and
the second equation is by I.H. on $\Mean{\bm{x_2}}$.

\smallskip

\noindent (cut-rule) \\
By the rule,
\begin{align*}
 {\sf ex}(\sigma, \bm{x})
= \TR{(\operatorname{Id}_1 \otimes s_{U,U} \otimes \operatorname{Id}_2)
\Comp ( \excon_1 \otimes \excon_2) 
}{U^2}{U^n}{U^n}
\end{align*}
We assume without loss of generality
that the $U$ for the $0_U$
of the
$\operatorname{Id} \otimes \, 0_U$ 
is a factor from the (co)domain 
of $\Mean{\bm{x}_1}$. 
Then,
LHS of the assertion is equal to \\
$\begin{array}{llr}
&
(\operatorname{Id} \otimes \, 0_U)  \Comp 
\TR{(\operatorname{Id}_1 \otimes s_{U,U} \otimes \operatorname{Id}_2)
\Comp ( \excon_1^\delta \otimes \excon_2^\delta) 
}{U^2}{U^n}{U^n}
\Comp (\operatorname{Id} \otimes \, 0_U)  \\
& =  \TR{(\operatorname{Id} \otimes \, 0_U) 
\Comp
(
(\operatorname{Id_1} \otimes s_{U,U} \otimes \operatorname{Id_2}) 
\Comp
(\excon_1^\delta \otimes \excon_2^\delta)
) \Comp (\operatorname{Id} \otimes \, 0_U) 
}{U^2}{U^n}{U^n} &  \text{naturalities} \\
&  =
\TR{(\operatorname{Id_1} \otimes s_{U,U} \otimes \operatorname{Id_2}) 
\Comp
(
(\operatorname{Id} \otimes \, 0_{U}) 
\Comp
\,  \excon_1^\delta \, \Comp (\operatorname{Id} \otimes \, 0_U))
\otimes  \excon_2^\delta) 
}{U^2}{U^n}{U^n}   &  \text{by the asm.} \\
& = \TR{(\operatorname{Id_1} \otimes s_{U,U} \otimes \operatorname{Id_2}) 
\Comp
(0_{U^{n_1+1}} \otimes \excon_2^\delta)
}{U^2}{U^n}{U^n} 
& 
 \text{I.H. on $\Mean{\bm{x_1}}$}  \\
&  =
\TR{(\operatorname{Id_1} \otimes
( 0_{U} \otimes U) \Comp s_{U,U} \Comp ( 0_{U} \otimes U)
\otimes \operatorname{Id_2})
\Comp (0_{U^{n_1+1}} \otimes \excon_2^\delta)
}{U^2}{U^n}{U^n} 
&   \text{
dinaturality}
\end{array}$ \\
$\begin{array}{llr}
& =
\TR{0_{U^{n_1+1}}}{U}{U^{n_1}}{U^{n_1}} 
\otimes
\TR{
(0_U \otimes \operatorname{Id}) \Comp \, \excon_2^\delta
}{U}{U^{n_2}}{U^{n_2}} 
&   \text{(\ref{prps}) and superposing} \\
& = \TR{0_{U^{n_1+1}}}{U}{U^{n_1}}{U^{n_1}} 
\otimes
\TR{
(0_U \otimes \operatorname{Id}) \Comp \, \excon_2^\delta
\Comp (0_U \otimes \operatorname{Id})
}{U}{U^{n_2}}{U^{n_2}} 
 &  \text{dinaturality} \\
& =
0_{U^{n_1}}
\otimes
0_{U^{n_2}}
&  \text{I.H. on $\Mean{\bm{x}_2}$} 
\end{array}$ 
\smallskip

\noindent 
The first dinaturality is via the decomposition
$0_{U^{n_1+1}}= 
(U^{n_1} \otimes 0_U)\Comp
\, 0_{U^{n_1}}
\Comp \,
(U^{n_1} \otimes 0_U)$
and the second dinaturality is via the decomposition $0_U = 0_U \Comp \, 0_U$.

\smallskip

\noindent ($\parr$-rule and additives) \\ 
Direct from the construction. 

\noindent See Figure \ref{picprf}
for a pictorial proof depicting the above rewriting in each case.
\end{proof}

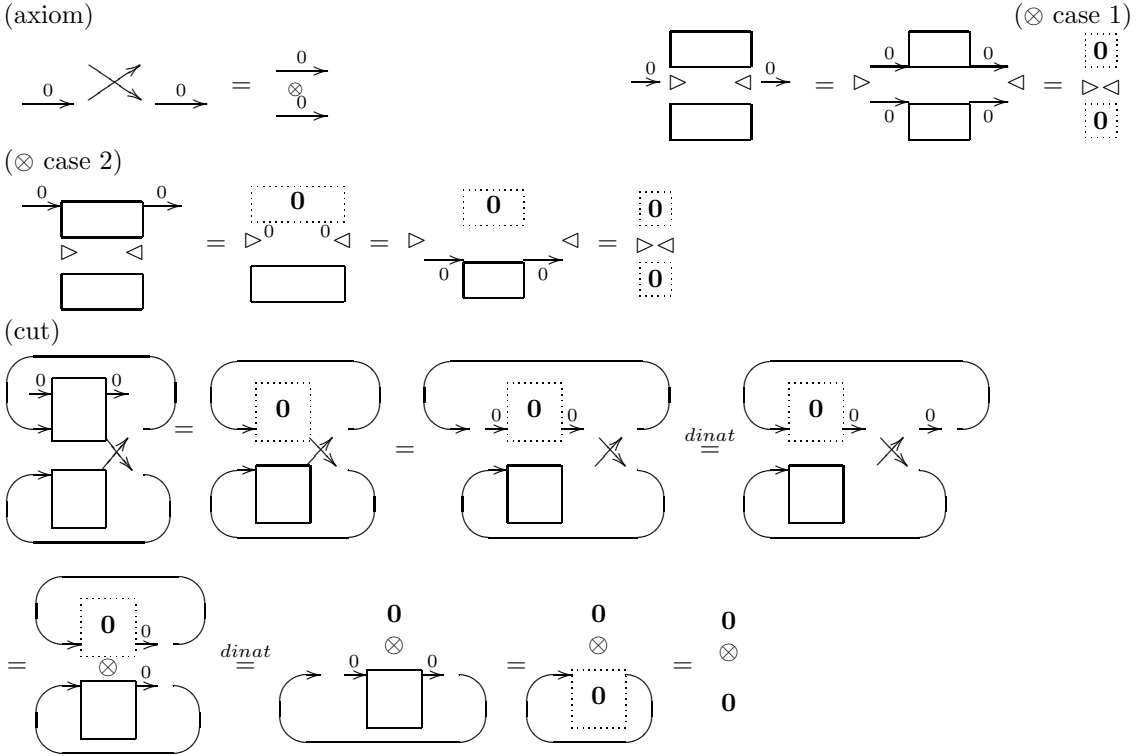
\begin{figure}[!ht]
\noindent (axiom) 
\hfill ($\otimes$ case 1)

\begin{minipage}{0.3\hsize}
\xymatrix@C=1.6pc@R=.9pc
{ &    \ar[rd]   &   &  \\
\ar[r]^0 & \ar[ru]
  &     \ar[r]^0 &     
   }
\end{minipage}
=
\begin{minipage}{0.3\hsize}
\xymatrix@C=1.6pc@R=.9pc
{ 
\ar[r]^0 \ar@{}[dr]|(.4){\otimes} &    \\
\ar[r]^0 &     
   }
\end{minipage}
\hfill \begin{minipage}{0.3\hsize}
\xymatrix@C=.9pc@R=.1pc{ 
&  & & \\
& & & \\
\ar[r]^0 & 
 \rhd
   &  
\lhd \ar[r]^0 
&    \\ 
    &   &    &  \\   
     &  &    &
\save"1,2"."2,3"*[F]\frm{}\restore   
\save"4,2"."5,3"*[F]\frm{}\restore 
}
\end{minipage}
= 
\begin{minipage}{0.3\hsize}
\xymatrix@C=.9pc@R=.1pc{ 
&  & & \\
\ar@<-1ex>[r]^0 & & \ar@<-1ex>[r]^0 & \\
\rhd & 
&   & \lhd  
   \\ 
\ar@<.6ex>[r]_0    &   &  \ar@<.6ex>[r]_0   &  \\   
     &  &    &
\save"1,2"."2,3"*[F]\frm{}\restore   
\save"4,2"."5,3"*[F]\frm{}\restore 
}
\end{minipage}
=
\begin{minipage}{0.3\hsize}
\xymatrix@C=1pc@R=.1pc{ 
\bm{0} \\
\rhd \lhd \\
\bm{0}
\save"1,1"."1,1"*[F.]\frm{}\restore   
\save"3,1"."3,1"*[F.]\frm{}\restore   
}
\end{minipage}

\smallskip

\noindent ($\otimes$ case 2)

\begin{minipage}{0.3\hsize}
\xymatrix@C=.9pc@R=.1pc{ 
\ar[r]^0 &  &  \ar[r]^0 & \\
& & & \\
 & 
 \rhd
& \lhd  &  
   \\ 
    &   &    &  \\   
     &  &    &
\save"1,2"."2,3"*[F]\frm{}\restore   
\save"4,2"."5,3"*[F]\frm{}\restore 
}
\end{minipage}
= 
\begin{minipage}{0.3\hsize}
\xymatrix@C=.9pc@R=.1pc{ 
 \ar@{}[dr] |{\txt{\normalsize $\bm{0}$}}   &    \\
 &  \\ 
 \rhd^0
& ^0\lhd \\
& \\
& 
\save"1,1"."2,2"*[F.]\frm{}\restore   
\save"4,1"."5,2"*[F]\frm{}\restore 
}
\end{minipage}
= 
\begin{minipage}{0.3\hsize}
\xymatrix@C=.9pc@R=.1pc{ 
&  \ar@{}[dr] |{\txt{\normalsize $\bm{0}$}}  & & \\
& & 
& \\
\rhd & 
&   & \lhd  
   \\ 
\ar@<.6ex>[r]_0    &   &  \ar@<.6ex>[r]_0   &  \\   
     &  &    &
\save"1,2"."2,3"*[F.]\frm{}\restore   
\save"4,2"."5,3"*[F]\frm{}\restore 
}
\end{minipage}
=
\begin{minipage}{0.3\hsize}
\xymatrix@C=1pc@R=.1pc{ 
\bm{0} \\
\rhd \lhd \\
\bm{0}
\save"1,1"."1,1"*[F.]\frm{}\restore   
\save"3,1"."3,1"*[F.]\frm{}\restore   
}
\end{minipage}
\smallskip

\noindent (cut) 
\smallskip

\noindent
\begin{minipage}{0.2\hsize}
\xymatrix@C=.7pc@R=.8pc
{  \ar@<-1ex>[r]^0 &    &        \ar@<-1ex>[r]^0 &    &   \\
&  \ar@{<-} `l []-<2em,0cm> `[u]+<0cm,1em> `[rr]+<1.5em,0cm> `[l] [rr] 
 &       \ar[rd]  &    &       =   \\ 
&     &   \ar[ru]  &
\ar `r []-<1.5em,0cm> `[d]-<0cm,1em> `[ll]-<2em,0cm>   `[r] [ll] 
      &   \\
 &   &     &   & 
\save"1,2"."2,3"*[F]\frm{}\restore 
\save"3,2"."4,3"*[F]\frm{}\restore 
}
\end{minipage}
\begin{minipage}{0.2\hsize}
\xymatrix@C=.7pc@R=.8pc
{    & \ar@{}[dr] |{\txt{\normalsize $\bm{0}$}}    &          &    &   \\
&  \ar@{<-} `l []-<2em,0cm> `[u]+<0cm,1em> `[rr]+<1.5em,0cm> `[l] [rr] 
 &        \ar[rd]  &    &          \\
&     &   \ar[ru]  &
\ar `r []-<1.5em,0cm> `[d]-<0cm,1em> `[ll]-<2em,0cm>   `[r] [ll] 
      &   \\
 &   &     &   & 
\save"1,2"."2,3"*[F.]\frm{}\restore 
\save"3,2"."4,3"*[F]\frm{}\restore 
}
\end{minipage}
=
\begin{minipage}{0.2\hsize}
\xymatrix@C=.7pc@R=.8pc
{&     & \ar@{}[dr] |{\txt{\normalsize $\bm{0}$}}    &
      &  &   \\
& \ar[r]^0  
\ar@{<-} `l []-<2em,0cm> `[u]+<0cm,1em> `[rrrr]+<1.5em,0cm>
 `[l] [rrrr]
& &  
 \ar[r]^0    &    \ar[rd]  &    &          \\
& &     &    & \ar[ru]  &
\ar `r []-<1.5em,0cm> `[d]-<0cm,1em> `[lll]-<2em,0cm>   `[r] [lll] 
      &   \\
&  &   &    &   &   & 
\save"1,3"."2,4"*[F.]\frm{}\restore 
\save"3,3"."4,4"*[F]\frm{}\restore 
}
\end{minipage}
$\stackrel{dinat}{=}$
\begin{minipage}{0.2\hsize}
\xymatrix@C=.7pc@R=.8pc
{     & \ar@{}[dr] |{\txt{\normalsize $\bm{0}$}}    &
      &  &   \\
   &  \ar@{<-} `l []-<2em,0cm> `[u]+<0cm,1em> `[rrrr]+<1.5em,0cm>
 `[l] [rrrr]  & 
 \ar[r]^0    &    \ar[rd]  &  \ar[r]^0  &          \\
&     &    & \ar[ru]  &
\ar `r []-<1.5em,0cm> `[d]-<0cm,1em> `[lll]-<2em,0cm>   `[r] [lll] 
      &   \\
 &   &    &   &   & 
\save"1,2"."2,3"*[F.]\frm{}\restore 
\save"3,2"."4,3"*[F]\frm{}\restore 
}
\end{minipage}

\bigskip

\noindent =
\begin{minipage}{0.2\hsize}
\xymatrix@C=.7pc@R=.8pc
{     & \ar@{}[dr] |{\txt{\normalsize $\bm{0}$}}    &
      &  &   \\
 & \ar@<-2ex>@{}[r] |(.5){\txt{\normalsize $\otimes$}}
  \ar@{<-} `l []-<2em,0cm> `[u]+<0cm,1em> `[rr]+<1.5em,0cm> `[l] [rr] 
 &        \ar[r]^0  &    &          \\
&     &   \ar[r]^0  &
\ar `r []-<1.5em,0cm> `[d]-<0cm,1em> `[ll]-<2em,0cm>   `[r] [ll] 
      &   \\
 &   &     &   & 
\save"1,2"."2,3"*[F.]\frm{}\restore 
\save"3,2"."4,3"*[F]\frm{}\restore 
}
\end{minipage}
$\stackrel{dinat}{=}$
\begin{minipage}{0.3\hsize}
\xymatrix@C=.7pc@R=.8pc
{     & \ar@{}[dr] |{\txt{\normalsize $\bm{0}$}}    &
      &  &   \\
 &
\ar@<-1ex>@{}[r] |(.5){\txt{\normalsize $\otimes$}}
 &       
 &    &          \\\
\ar[r]^0  &     &   \ar[r]^0  &
\ar `r []-<1.5em,0cm> `[d]-<0cm,1em> `[lll]-<2em,0cm>   `[r] [lll] 
      &   \\
 &   &     &   & 
\save"3,2"."4,3"*[F]\frm{}\restore 
}
\end{minipage}
=\begin{minipage}{0.3\hsize}
\xymatrix@C=.7pc@R=.8pc
{     & \ar@{}[dr] |{\txt{\normalsize $\bm{0}$}}    &
      &     \\
 & \ar@<-1ex>@{}[r] |(.5){\txt{\normalsize $\otimes$}}
 &        
&             \\
&  \ar@{}[dr] |{\txt{\normalsize $\bm{0}$}}     &   
  \ar `r []-<1.5em,0cm> `[d]-<0cm,1em> `[l]-<2em,0cm>   `[r] [l] 
  &      \\
 &   &     &   
\save"3,2"."4,3"*[F.]\frm{}\restore 
}
\end{minipage}
=
\begin{minipage}{0.3\hsize}
\xymatrix@C=.7pc@R=.8pc
{  \ar@{}[dr] |{\txt{\normalsize $\bm{0}$}}    &
         \\
  \ar@<-1ex>@{}[r] |(.5){\txt{\normalsize $\otimes$}}
 &       
          \\
  \ar@{}[dr] |{\txt{\normalsize $\bm{0}$}}     &   \\
 & }
\end{minipage}

\smallskip
\noindent
In all the cases above, the upper (resp. lower) square
represents $\excon_1$ (resp. $\excon_2$), and 
each $\bm{0}$ inside the dotted square arises by an induction
hypothesis.
\caption{Pictorial Proof of Lemma \ref{sublem}}
\label{picprf}
\end{figure}

\begin{prop}[Mismatch gives rise to
zero convergence of Ex] ~\\ \label{convzero}
For two $\MALL$ proofs $\pf{\pi^i}{\Delta_i}{\Gamma_i,A_i}$
with $i=1,2$, let $\bm{x}_i = \lambda_i \times  (\gamma_i, a_i) 
\in \abs{\pf{\pi^i}{\Delta_i}{\Gamma_i, A_i}}$
so that 
$a_i \in \abs{A_i}$ with $a_1 \not = a_2$
and $A_1$ and $A_2$ are pairwise dual formulas.
Then
\begin{align*}
&
 \Ex{\sigma_{a_1,a_2}}{
\Ex{\sigma
}{\bm{x}_1 }
\otimes
\Ex{\sigma
}{\bm{x}_2 }
} = 0_{U^{n_1+n_2}} \hfill & &  \text{where $\sigma_{a_1,a_2} =
s^{\delta_{a_1,a_2}}$}
\end{align*}
\end{prop}
\smallskip
Note that the LHS of the assertion is, by Definition \ref{Exx},
the following, in which
$\epsilon$ is the action arising from $\sigma_{a_1,a_2}
=s_{U_{a_1},U_{a_2}}^0$ of Definition
 \ref{zeroret}:
\begin{align} 
&  \TR{ (\operatorname{Id} \otimes \sigma_{a_1,a_2}
\otimes \sigma_{\bm{x}_1} \otimes \sigma_{\bm{x}_2})
\Comp (\Mean{\bm{x}_1}^{\zr{\bm{x}_1}} \otimes
\Mean{\bm{x}_2}^{\zr{\bm{x}_1}})^\epsilon
}{U^{2(1+m_1 +m_2)}}{U^{n_1 + n_2}}{U^{n_1 + n_2}} \label{eqprf} \\
& =
\TR{(\operatorname{Id} \otimes \sigma_{a_1,a_2})
\Comp ((\Ex{\sigma}{\bm{x}_1} \otimes
 \Ex{\sigma}{\bm{x}_2})^\epsilon
 }{U^{2(1+m_1 +m_2)}}{U^{n_1 + n_2}}{U^{n_1 + n_2}} \quad 
\txt{by nats and vanish IIs}  \nonumber
\end{align}
In the above $\Mean{\bm{x}_j}$ is an endomorphism on $U^{2(m_j+1)+n_j}$
with the subfactor $U^{2m_j}$ for $\sigma_{\bm{x}_j}$.


\smallskip

Before the proof of Proposition \ref{convzero},
let us observe a general equation derivable from certain trace axioms
(dinaturality and yanking), where $f: X \otimes U \longrightarrow Y
\otimes U$
and 
$0_{U,I}$ (resp. $0_{I,U}$) is the zero morphism from $U$
to the tensor unit $I$ (resp. the other way around).
\begin{align} \label{vanishtr}
\TR{(\operatorname{Id} \otimes \, 0_{U}) \Comp f}{U}{X}{Y} & = 
 (\operatorname{Id} \otimes \, 0_{U,I}) \Comp  f
\Comp \, (\operatorname{Id} \otimes \, 0_{I,U})
\end{align}
Note first that the zero morphisms $0_{U,I}$ and $0_{I,U}$ above are derivable from
$0_U$ using the trace: \\
$\begin{aligned}
0_{U,I} = \TR{U \otimes U \stt{0_U \otimes 0_U} U \otimes U 
\stt{j} U \cong I \otimes U }{U}{U}{I} \\
 0_{I, U} = \TR{I \otimes U \cong U \stt{k}   U \otimes U 
 \stt{0_U \otimes 0_U}  U \otimes U   }{U}{U}{I}
\end{aligned}$

See Figure \ref{picconvzero} (lower-right) depicting
Equation (\ref{vanishtr}).

\smallskip

\noindent (proof of (\ref{vanishtr})) \\
By the decomposition $0_U=0_{I,U} \Comp \, 0_{U,I}$,
the LHS is
$\TR{(\operatorname{Id} \otimes \, 0_{U,I}) \Comp f
\Comp (\operatorname{Id} \otimes \, 0_{I,U})
}{I}{X \otimes I}{Y \otimes I}$, by dinaturality,
which is equal to
the RHS by vanishing.
 \hfill (end of proof of (\ref{vanishtr}))

\smallskip

Finally we go to:
\begin{proof}{}[Proof of Proposition \ref{convzero}]
We prove the following instance of the proposition 
using Equation (\ref{eqprf}),
where $n=n_1+n_2$, 
since $\sigma_{a_1,a_2}=s^0= 0_U^2$:
\begin{align} 
\TR{( \operatorname{Id} \otimes \, 0_U^2  ) \Comp 
( \Excon_1
\otimes
\Excon_2)^\epsilon
}{U^2}{U^n}{U^n} =0_{U^{n_1+n_2}} & &  
\txt{where $\Excon_i = \Ex{\sigma}{\bm{x}_i}$
 }  \nonumber
\intertext{For this, it suffices to show the following stronger equation,
as $\excon_i$ is $\Excon_i$ ridden of the zero action 
on the assoc-rets and the assoc-corets (cf. Definition \ref{zeroret}):}
\TR{(\operatorname{Id} \otimes \, 0_U^2 ) \Comp 
( \excon_1
\otimes
\excon_2)^\epsilon
}{U^2}{U^n}{U^n} =0_{U^{n_1+n_2}} & & 
\txt{where $\excon_i = \ex{\sigma}{\bm{x}_i}$} \nonumber
\intertext{By superposing (after the distribution of $\epsilon$ over $\otimes$),
the LHS is equal to}
\TR{(\operatorname{Id} \otimes \, 0_U ) \Comp \, \excon_1^\epsilon}{U}{U^{n_1}}{U^{n_1}}
 \otimes
\TR{ (\operatorname{Id} \otimes \, 0_U) \Comp \, {\sf
 ex}_2^\epsilon}{U}{U^{n_2}}{U^{n_2}} \label{tentrace}
\end{align}
On the other hand, Lemma \ref{sublem} and (\ref{vanishtr}) say
for all $i=1,2$
\begin{align*}
\TR{(0_U \! \otimes \operatorname{Id}) \Comp \, {\sf
 ex}_i^\delta}{U}{U^{n_i}}{U^{n_i}}= 0_{U^{n_i}}
\end{align*}
Since the two actions $\epsilon$ and $\delta$ coincide again by
(\ref{vanishtr}), 
the formula (\ref{tentrace}) becomes equal to $0_{U^{n_1}} \otimes 0_{U^{n_2}}=0_{U^{n}}$.
\end{proof}


\subsection{Main Theorem}
This section concerns the main theorem of this paper.

\begin{theorem}[Ex is invariant and diminishes sets of indices] \label{mainthm}~~\\
Let $\nu \in \abs{\pf{\pi}{\Delta}{\Gamma}}^J
\indrd \nu' \in \abs{\pf{\pi'}{\Delta'}{\Gamma}}^{J'}
$ be any $\MALLIC$ proof transformation. 
Then
\begin{enumerate}
\item[(i)]
\begin{align*}
\res{\inEx{\nu}{J}}{J'} = \inEx{\nu'}{J'} \quad &
\quad  \text{and} \quad  &  \quad
\forall j \in J \setminus J' \quad
\ptEx{\nu_j}= 0 \nonumber
\end{align*}

\item[(ii)] In particular, when $\pi'$ is cut-free so that $\Delta'$
is empty, then
\begin{align*}
\res{\inEx{\nu}{J}}{J'} =\Mean{\nu'} 
\quad &
\quad  \text{and} \quad 
& J'  = \{ j \in J \mid \ptEx{\nu_j} \not = 0 \, \,
 \} 
\end{align*}

\end{enumerate}

\end{theorem}

\begin{proof}{}
 We prove (i) according to the cases of
 Proposition \ref{ltit}, since (ii) follows directly from (i) as follows: 
For a cut-free $\pi'$, $\sigma_{\nu'}$ is empty, hence
$\ptEx{\nu_j}=\Mean{\nu'_j}$, where $\Mean{\nu'_j} \not = 0$
holds directly
both from
the construction of Definition \ref{endopartia}
 and from the non-collapsing assumption (\ref{noncollapse}).

 \smallskip
The invariance of (i) is direct 
 by the yanking axiom in case 1, and by induction on the proof $\pi$
 in cases 2 and 3: This proof method directly comes as
an instance of the known method in the symmetric traced
 monoidal category modelling multiplicative GoI
\cite{HS06}.
Thus we prove the zero convergence
 for the diminution of $J$.
The following crucial cases are those of the proof of Proposition \ref{ltit}.

\smallskip

\noindent (Crucial case 1) \\
Each instance of $\nu$ at $j \in J \setminus J'$
is
$\nu_j=(\delta_j,b_j,a_j,b_j,\lambda_j)$, 
so that $a_j \not = b_j$ since $j \not \in J'$, \\ then $\ptEx{\nu_j}=0$
by Proposition \ref{convzero}.

\smallskip

 \noindent (Crucial case 2) \\
 $J=J_1 + J_2$ diminishes into $J_1$.
Each instance of $\tau$ at $j_2 \in J_2$
is $\tau_{j_2}=(\omega, \delta_2, \gamma, (2,a_2))
\in \abs{\&(\pi^1,\pi^2)}$.
Thus each instance of $\nu$ at $j_2 \in J_2$
is $\nu_{j_2}=
(\omega, \delta_2,\delta_3, (2,a_2), (1,a_1), \gamma, \xi)
\in \abs{\pi}$.
Since $(2,a_2) \not =  (1,a_1)$,
we have $\ptEx{\nu_{j_2}}=0$
by Proposition \ref{convzero}.

\smallskip

 \noindent (Crucial case 3) \\
 $J$ does not diminish in this case. \end{proof}

\smallskip

Appendix \ref{prolo} is read as
an elucidating example of Theorem \ref{mainthm}.

\bigskip

\noindent {\bf \Large  Conclusion and Future Work} \\
This paper offers two main contributions:
 \begin{enumerate}
  \item[(i)]
       Presenting an indexed $\MALL$ system for stacking cut formulas
       and its relational counterpart to
       simulate $\MALL$ proof reduction of cut elimination.
  \item[(ii)] Constructing an execution formula for the interpretation
	      of $\MALL$ proofs equipped with indices.
	      The $\MALL$ proof reduction is characterised
	      by the convergence of the execution formula
	      to the denotational interpretation.
       Furthermore, the zero convergence of the execution formula
	      characterises the diminution of indices, which is specific
	      to additive cut elimination.
 \end{enumerate}

Our explicit use of indexed-syntactical manipulations
directly overcomes known difficulties in additive GoI.
We hope that this paper, from the perspective of indexed linear logic,
will shed light on an approachable understanding
of the preceding literature on
additive GoI, from precursory ones \cite{Gi95, Duchesne}
to more recent developments \cite{Gi11, Seil}.

We discuss some future directions. \\
For a genuine $\MALL$ GoI without
bypassing via indexed logic,
a syntax-free counterpart is required  
to replace the indices.
We construct such a genuine GoI \cite{Ham}
using an algebraic ingredient: a scalar extension
of Girard's $*$-algebra of partial isometries
over a boolean polynomial semi-ring.
The genuine GoI may help us connect
our syntactic manipulation of indices
to Girard's semantic use of clauses for predicates
in the precursory $\MALL$ GoI \cite{Gi95}.

In a syntactic direction, the status of Gentzen cut elimination
for $\MALLI$ remains open since the present paper only concerns
lifting the image to the indices of $\MALL$ cut-reduction.
The status will complement the 
reduction-free cut elimination, known to be derivable
from the Fundamental lemma \ref{FunLem} 
(cf.\ \cite{BE,BE2,HamTak08}).


Extending the present paper to the exponentials is challenging to
use Bucciarelli--Ehrhard $\LLI$ \cite{BE2} for modelling GoI.
This will involve extending our methodology of a traced monoidal
category with
a zero morphism to the whole GoI situation \cite{HS06,HS11},
compatibly with the multisets interpretation of
the exponential connective $\!$
in $\Rel$.
The explicit accommodation of the indices to GoI 
will give a novel approach to the (non indexed)
GoI modelling for the exponentials.

\bigskip

\noindent
{\bf Acknowledgment}
We wish to thank the referees for detailed and very helpful comments that have
greatly improved the presentation.

\appendix

\section{Axioms of Traced Monoidal Category}

\begin{definition}[Trace axioms of the family
${\sf Tr}^Z_{X,Y}$ of
 (\ref{trfam})
\cite{JSV96, HS06}] \label{traceaxioms}

\begin{enumerate}
 \item  ({\bf Natural} in $X$) \\
$\begin{aligned}
\TR{f}{Z}{X}{Y} \Comp g =
\TR{f \Comp (g \otimes Z)}{Z}{X'}{Y}
&  & \mbox{where} \quad f : X \otimes Z \longrightarrow Y \otimes Z
\quad \mbox{and} \quad
g : X' \rightarrow X 
\end{aligned}$

\item ({\bf Natural} in $Y$) \\
$\begin{aligned}
g \Comp \TR{f}{Z}{X}{Y} =
\TR{ (g \otimes Z) \Comp f }{Z}{X'}{Y}
& &  \mbox{where} \quad f : X \otimes Z \longrightarrow Y \otimes Z
\quad \mbox{and} \quad
g :  Y \rightarrow Y'
\end{aligned}$

\item ({\bf Dinatural} in $Z$) \\
$\begin{aligned}
\TR{ (Y \otimes g) \Comp f }{Z}{X}{Y}=
\TR{f \Comp ( X \otimes g)  }{Z'}{X}{Y}
&& \mbox{where} \quad f : X \otimes Z \longrightarrow Y \otimes Z
\quad \mbox{and} \quad
g :  Z' \rightarrow Z
\end{aligned}$

\item ({\bf Vanishing I} ) \\
$\begin{aligned}
\TR{f \otimes I}{I}{X}{Y}= f & & \mbox{where} \quad f: X 
\longrightarrow 
Y 
\end{aligned}$

\item ({\bf Vanishing II} ) \\
$\begin{aligned}
\TR{f}{Z \otimes W}{X}{Y} =
\TR{\TR{f}{W}{X \otimes Z}{Y \otimes Z}}{Z}{X}{Y}
& & \mbox{where} \quad f: X \otimes Z \otimes W \longrightarrow
 Y \otimes Z \otimes W
\end{aligned}$

\item ({\bf Superposing}) \\
$\begin{aligned}
g \otimes 
\TR{f}{Z}{X}{Y} =
\TR{g \otimes f}{Z}{W \otimes X}{V \otimes Y}
& & \mbox{where} \quad f : X \otimes Z \longrightarrow Y \otimes Z
\quad \mbox{and} \quad
g: W \longrightarrow V
\end{aligned}$

\item ({\bf Yanking}) \\
$\begin{aligned}
\TR{s_{X,X}}{X}{X}{X}=X
& &
\mbox{for the symmetry} \quad s_{X,X} : X \otimes X \longrightarrow X \otimes X
\end{aligned}$
\end{enumerate}

\end{definition}

\begin{lemma}[Generalized Yanking \cite{HS11}] 
Let $s_{Z,Y}$ denote the symmetry from $Z \otimes Y$ to $Y \otimes Z$.
\begin{align*}
\TR{s_{Z,Y} \Comp (f \otimes g)}{Z}{X}{Y}= g \Comp f 
&&
\mbox{where} \quad f: X \longrightarrow Z \quad \mbox{and} \quad
g: Z \longrightarrow Y 
\end{align*}
\end{lemma}
\begin{proof}{}
$LHS \stackrel{nat}{=}
\TR{s_{Z,Y} \Comp (Z \otimes g)}{Z}{Z}{Y} \Comp f
\stackrel{dinat}{=}
\TR{( Y \otimes g) \Comp s_{Z,Y}}{Y}{Z}{Y} \Comp f$. Inside the trace 
$( Y \otimes g) \Comp s_{Z,Y}
= s_{Y,Y} \Comp (g \otimes Y)$, thus
$\TR{s_{Y,Y} \Comp (g \otimes Y)}{Y}{Z}{Y}
\stackrel{nat}=
\TR{s_{Y,Y}}{Y}{Y}{Y} \Comp  g
\stackrel{yank}{=}g$
\end{proof}
\section{Omitted Proofs}
\subsection{Proof for Proposition \ref{FunLem} (Fundamental Lemma)}
\label{apFunLem}
\begin{lemma}[(i) implies (ii)] \label{iimpii}
Let $\pf{\pi}{\Delta}{\Gamma}$ be a proof of a sequent
$\vdash [\Delta], \, \, \Gamma$ in $\MALLC$.
Let $\delta \times \gamma \in \abs{\pf{\pi}{\Delta}{\Gamma}}^J$
(for some $J \subseteq I$)
with $\delta \in \sblabs{\sbl{\Delta}}^J$ and $\gamma \in \abs{\Gamma}^J$.
The sequent $\inseqmeansepare{J}{\Delta}{\Gamma}{\delta}{\gamma}$
has a proof $\rho$ in $\MALLIC$
such that $\res{\rho}{\emptyset}=\pi$.
\end{lemma}
\begin{proof}{}
By construction on the $\MALL$ proof $\pi$.
The proof figures are referred in Definition \ref{intRelc}.

\noindent (cut rule) \\
$\delta \times \gamma \cong 
\delta_1 \times \delta_2 \times
\gamma_1 \times \gamma_2$
with $\delta_1 \times \gamma_1 \in \pf{\pi^1}{\Delta_1}{\Gamma_1,A}$
and $\delta_2 \times \gamma_2 \in \pf{\pi^2}{\Delta_2}{A^\perp, \Gamma_2}$.
By I.H.s on $(\delta_i \times \gamma_i) \,$s,
there are $\MALLI$-proofs of
the sequents 
$\inseqmeanseparethree{J}{\Delta_1}{\Gamma_1}{A}{\delta_1}{\gamma_1}$
and 
$\inseqmeanseparethree{J}{\Delta_2}{A^\perp}{\Gamma_2}{\delta_2}{\gamma_2}$
with $\gamma_i = \gamma_i^{'}
\times \gamma_i^{''}$.
Note that $\formmean{A}{\gamma_1^{''}}$
and $\formmean{A^\perp}{\gamma_2^{'}}$
are dual formulas since they have the same domain $J$.
Hence the cut between the dual formulas is applied to prove
$\inseq{J}{\, \, 
\sblformmean{\Delta_1}{\delta_1},
\sblformmean{\Delta_2}{\delta_2},
\formmean{A}{\gamma_1^{''}},
\formmean{A^\perp}{\gamma_2^{'}}
\, \, }{\formmean{\Gamma_1}{\gamma_1^{'}},
\formmean{\Gamma_2}{\gamma_2^{''}}}$.
The assertion follows since
$\sblformmean{\Delta_1}{\delta_1}
\times
\sblformmean{\Delta_2}{\delta_2}
=\sblformmean{(\Delta_1, \Delta_2)}{\delta_1
\times \delta_2}$.

\noindent ($\&$-rule) \\
Let $\nu=\abs{\pf{\pi}{\Delta}{\Gamma}}$. Then
$\nu = 
\{ (x_1, z, y, (1, a_1)) 
\mid 
(x_1, z, y, a_1 ) 
\in \nu_1 \}
\, + \, 
\{ (x_2, z, y, (2, a_2)) 
\mid 
(x_2, z, y, a_2) 
\in \nu_2 \}
\, \cong \, \merge{\nu_1}{\nu_2}
$
with $\gamma_i \in \abs{\pf{\pi^i}{\Delta_i, \Delta}{\Gamma,A_i}}^{J_1}$
and $J=J_1 + J_2$.
By I.H.s on $\pi^i\,$s, there are $\MALLI$-proofs of
$\inseq{J_i}{\, 
\sblformmean{\Delta_i}{\delta_i^{'}},
\sblformmean{\Delta}{\delta_i^{''}}
\, }{
\formmean{\Gamma}{\gamma_i^{'}},
\formmean{A_i}{\gamma_i^{''}}}
$
with $\delta_i = \delta_i^{'} \times
\delta_i^{''}$
and $\gamma_i =  \gamma_i^{'}
\times \gamma_i^{''}$.
Because 
$\res{\formmean{\Gamma}{
\merge{
\gamma_1^{'}
}{
\gamma_2^{'}
}}}{J_i}
=
\formmean{\Gamma}{
\gamma_i^{'}
}$
and $\res{\sblformmean{\Delta}{
\merge{
\delta_1^{''}}{
\delta_2^{''}}}
}{J_i}
=
\sblformmean{\Delta}{
\delta_i^{''}}
$
by Lemma \ref{respro},
the $\&$-rule is applied to prove \\
$\inseq{J_1 + J_2}{ \, \, 
\sblformmean{\Delta_1}{\delta_1^{'}},
\sblformmean{\Delta_2}{\delta_2^{'}},
\sblformmean{\Delta}{
\merge{
\delta_1^{''}}{
\delta_2^{''}}
} \, \, }
{
\formmean{\Gamma}{
\merge{
\gamma_1^{'}
}{
\gamma_2^{'}
}},
\, \formmean{A_1}{\gamma_1^{''}}
\, \& \, 
\formmean{A_2}{\gamma_2^{''}}
}
$.

\end{proof}

\begin{lemma}[(ii) implies (i)] \label{iiimpi}
Let $\vdash [\Delta], \, \, \Gamma$ be a sequent of $\MALLC$.
Let $\nu \in (\sbl{\Delta} \times \Gamma)^J$ (for some $J \subseteq
 I$)
and let $\rho$ be a proof of $\inseqmean{J}{\Delta}{\Gamma}{\nu}$
in $\MALLIC$.
Then $\nu \in \abs{\pf{(\res{\rho}{\emptyset})}{\Delta}{\Gamma}}^J$

\end{lemma}

\begin{proof}{}
By the construction on the $\MALLI$ proof $\rho$.

\noindent (cut rule)  \\ $\rho$ is
$\infer[\cut]{
\vdash_J [\Delta_1, \Delta_2, A, A^\perp],\, \, \Gamma_1, \Gamma_2}{
\deduce{\vdash_J
[\Delta_1], \, \, \Gamma_1,  A}{\rho^1} & 
\vdash_J
\deduce{[\Delta_2], \, \, A^{\perp}, \Gamma_2}{\rho^2}}
$. \\
The conclusion is written
$\inseqmean{J}{\Delta_1,\Delta_2,A,A^\perp}{\Gamma_1,\Gamma_2}{\nu}$.
From the construction, $\nu \cong \nu_1 \times \nu_2$
so that the conclusions of $\rho^1$ and $\rho_2$
are respectively
$\inseqmean{J}{\Delta_1}{\Gamma_1,A}{\nu_1}$
and $\inseqmean{J}{\Delta_2}{A^\perp, \Gamma_2}{\nu_2}$.
By I.H.s on $\rho^i$s,
$\nu_1 \in \abs{
\pf{(\res{\rho^1}{\emptyset})}{\Delta_1}{\Gamma_1,A}
}^{J}$
and
$\nu_2 \in \abs{
\pf{(\res{\rho^2}{\emptyset})}{\Delta_2}{A^\perp, \Gamma_2}
}^{J}$. The assertion follows since \\ 
$\abs{
\pf{(\res{\rho^1}{\emptyset})}{\Delta_1}{\Gamma_1,A}
}^{J}
\times
\abs{
\pf{(\res{\rho^2}{\emptyset})}{\Delta_2}{A^\perp, \Gamma_2}
}^{J}
\cong
\abs{
\pf{(\res{\rho}{\emptyset})}{\Delta_1, \Delta_2,A,A^\perp}{\Gamma_1,\Gamma_2}
}^{J}.
$

\smallskip

\noindent ($\&$-rule) \\ 
$\rho$ is 
$\infer[\&]{\vdash_{J_1 + J_2} [\Delta_1, \Delta_2, \Sigma \,], \, \, \Gamma, A_1 \& A_2}{
\deduce{\vdash_{J_1}  [\Delta_1, \Sigma \, ], \, \, \Gamma, A_1}{\rho^1}
& 
\deduce{\vdash_{J_2} [\Delta_2, \Sigma \, ], \, \, \Gamma,
 A_2}{\rho^2}}$ \\
$\nu$ is of the form $\merge{\nu_1}{\nu_2}$
so that the conclusions of $\rho_i\,$ s are
$\inseqmean{J_i}{\Delta_i,\Sigma}{\Gamma,A_i}{\nu_i}$.
By I.H.s on $\rho^i$s,
$\nu_i \in
\abs{
\pf{(\res{\rho^i}{\emptyset})}{\Delta_i, \Sigma}{\Gamma_1,A_i}
}^{J_i}$.
The assertion follows since  \\
$\merge{
\abs{
\pf{(\res{\rho^1}{\emptyset})}{\Delta_1,\Sigma}{\Gamma,A_1}
}^{J_1}
}{
\abs{
\pf{(\res{\rho^1}{\emptyset})}{\Delta_2,\Sigma}{\Gamma,A_2}
}^{J_2}
}
\cong
\abs{
\pf{(\res{\rho}{\emptyset})}{\Delta_1,\Delta_2, \Sigma}{\Gamma,A_1 \& A_2}
}^{J_1 + J_2}
$

\end{proof}

\section{Indices and Additive Cut Elimination  }\label{prolo}
This appendix elucidates the fundamental idea
of the paper.
The appendix may read as a prologue of the paper
by readers yet familiar with $\MLL$ GoI interpretation on
a reflexive object in a traced monoidal category.

\smallskip

Consider a sequence $\pi_1 \rhd \pi_2 \rhd \pi_3$
of cut eliminations for
proofs in the additive fragment of $\MALL$.
In our sequent notation, pairwise cut formulas, if present, 
 are stored inside a stack $[\cdot]$
in a sequent.
The first reduction, intrinsic to the additives,
eliminates a $\&$ in a cut, whereby the subproof $\ax_2$ is pruned.
The second reduction eliminates a redundant cut against an axiom:

\noindent
\begin{minipage}[lc]{3.2in}
$\infer[\cut]{\vdash [A \& A, A^\perp \oplus A^\perp] \, \,
A^\perp, A}{
\infer[\&]{\vdash A^\perp, A \& A}{
\infer[\ax_1]{\vdash A^\perp, A}{} & \infer[\ax_2]{\vdash A^\perp, A}{}}
&
\infer[\oplus_1]{\vdash A^\perp \oplus A^\perp, A}{
\infer[\ax_3]{\vdash A^\perp, A}{}}
}$
\end{minipage}
\begin{minipage}[c]{.2in}
$\rhd$
\end{minipage}
\begin{minipage}[c]{2in}
$\infer[\cut]{\vdash [A,A^\perp] \, \,  A^\perp, A}{
\infer[\ax_1]{\vdash A^\perp, A}{} & \infer[\ax_3]{\vdash A^\perp, A}{}}$
\end{minipage}
\begin{minipage}[c]{.2in}
\end{minipage}

\hspace{51ex}
$\rhd$ \quad
\begin{minipage}[c]{.7in}
$\infer[\ax_1]{\vdash A^\perp A}{}$
\end{minipage}

\vspace{2ex}

\noindent {\bf Step 1 (Interpretation $\abs{\pi}$ in $\Rel$
with unperformed cuts and indices for additives)} \\
We begin by interpreting proofs in $\Rel$
but without relational composition.
For this, the cut rule is interpreted same as the tensor rule.
This interpretation is consistent with the syntactic convention,
starting from Girard's GoI 1,
which puts the cut formulas into a stack.

For simplicity and in accordance with the fact
that the multiplicative dual elements 1 and $\perp$
are interpreted in $\Rel$ as the singleton set, we take
$A=1$ and, dually, $A^\perp=\>\perp$, 
so that 
$\abs{1}=\abs{\! \perp \!}$ is the singleton,
whose unique element is denoted $*$ or $\bar{*}$ (obviously, $*=\bar{*}$),
depending on whether it comes from $\abs{A}$ or $\abs{A^\perp}$, respectively.

An axiom is interpreted in $\Rel$ by the diagonal, so that 
$\abs{\ax_i}= \{ (\bar{*},*) \} \subseteq \abs{A^\perp} \times \abs{A}$.
The proof $\pi_2$ is interpreted as
$\abs{\textrm{cut}(\ax_1,\ax_2)}= \{ (\bar{*},*, \bar{*},*) \}
\subseteq 
\abs{A^\perp} \times  \, [ \, \abs{A}
\times 
\abs{A^\perp} \, ] \,  \times \abs{A}$,
in which the pair $(*,\bar{*})$ in the cut slot
from $[ \, \abs{A} \times \abs{A^\perp} \, ]$
remains explicit, rather than being hidden
by relational composition through $*=\bar{*}$. 
Note that both interpretations $\abs{\pi_2}$ and $\abs{\pi_3}$
are singletons. More generally, it is straightforward to see
that any proof in the multiplicatives can be interpreted by a singleton
whenever literals are interpreted by singletons.
However, this is not the case for the additives.
When interpreting $\pi_1$ with additive rules,
singletons prove insufficient, and this is where the indices become
necessary:
The left and right premises of $\pi_1$ are interpreted respectively by 
$\abs{\&(\ax_1,\ax_2)}=\{ (\bar{*},(1,*)), (\bar{*},(2,*))  \} 
\subseteq \abs{A^\perp} \times  (\abs{A} + \abs{A})
$ and 
$\abs{\oplus_1 \! (\ax_3)}=\{ ((1,\bar{*}),*)   \}
\subseteq (\abs{A^\perp} + \abs{A^\perp}) \times \abs{A}  
$.
A set of indices $J=\{1,2 \}$ is employed to describe
these two interpretations: 
the first yields
$\delta \in \abs{\&(\ax_1,\ax_2)}^J$
so that $\delta_1= (\bar{*},(1,*))$ and $\delta_2=  (\bar{*},(2,*))$, and
the second yields 
$\tau \in \abs{\oplus_1 \! (\ax_3)}^J$,
so $\tau_1=\tau_2=((1,\bar{*}),*)$.  \\
Then $\pi_1$ is interpreted by $\nu$:\\
$\begin{aligned}
 \nu :=\delta \times \tau \in \abs{\textrm{cut}(\&(\ax_1,\ax_2), \oplus_1(\ax_3))}^J
\subseteq 
(\abs{A^\perp} \times [ \, 
(\abs{A} \! + \! \abs{A}) \times
(\abs{A^\perp} \! + \! \abs{A^\perp}) \, ] 
\times \abs{A} )^J,
\end{aligned}$ \\
and therefore $(\delta \times \tau)_1 = (\delta_1, \tau_1) =
(\bar{*},(1,*),(1,\bar{*}),*)$
and $(\delta \times \tau)_2 = (\delta_2 , \tau_2) =
(\bar{*},(2,*),(1,\bar{*}),*)$,
where $\delta \times \tau$ denotes the mediating morphism of the
set-theoretical cartesian product.
Summing up, $\abs{\pi_1}=\{ v_j \mid j \in J=\{ 1,2 \}  \}$,
where $v_1= (\bar{*},(1,*),(1,\bar{*}),*)$ and 
$v_2= (\bar{*},(2,*),(1,\bar{*}),*)$.

\smallskip

\noindent {\bf Step 2 ($\inEx{\nu}{J}$ for 
$\nu \in \abs{\pi}^J$: Executing cuts using trace structures)} \\
In addition to step 1,
our GoI interpretation runs an execution formula
for $\abs{\pi_i}$ to perform
cut elimination against the unperformed cut formulas,
syntactically in the stack and semantically
in the noncompositional interpretation.

Each point in $\abs{\pi}$ is interpreted as an endomorphism
on a certain tensor folding of a \textit{reflexive object} $U$ in a traced
monoidal category ${\cal C}$ with a zero morphism on $U$.
The object $U$ uniformly interprets each element in the interpretation of
the conclusion of $\pi$; e.g.,
in $\abs{\pi_1}$, $U$ has the elements $\bar{*}$, $*$,
$(1,\bar{*})$, 
$(1,*)$, and $(2,*)$.
In the following, these points are identified with
their interpretation $U$.

For the most primitive case, e.g., for $\abs{\pi_3}$,
each diagonal point $(\bar{*},*)$ interpreting the axiom is interpreted as a \textit{symmetry}
of $\cal{C}$:
\begin{align}
s_{\bar{*},*} : U_{\bar{*}} \otimes U_{*} \longrightarrow 
U_{\bar{*}} \otimes U_{*} \label{ex3}
\end{align}
The unique point of $\abs{\pi_2}$ is interpreted by
the endomorphism $\Ex{\sigma_\textrm{cut}}{\abs{\pi_2}}$
on $U_{\bar{*}} \otimes U_{*}$, in which $\sigma_\textrm{cut}$,
as the interpretation of the cut,
is the symmetry $s_{*,\bar{*}}$ acting on the cut formulas:
\begin{alignat}{1}
  \Ex{\sigma_\textrm{cut}}{\abs{\pi_2}}=
   \TR{(\bar{*} \otimes \sigma_\textrm{cut} \otimes *) \Comp (s_{\bar{*},*}
   \otimes s_{\bar{*},*} )}{* \otimes \bar{*}}{\bar{*} \otimes
   *}{\bar{*} \otimes *}  \label{ex2}
 \end{alignat} 
Note that the symmetries $s$'s occurring in (\ref{ex2}) interpret
respective axioms.  \\
This is equal to (\ref{ex3}) in ${\cal C}$ by the trace axioms. 
The adjacent diagrams illustrate (\ref{ex3}) and (\ref{ex2}),
where the equality is found in the diagram for (\ref{ex2})
by chasing arrows with respect to both composition and feedback.

 \begin{wrapfigure}[14]{r}{11cm}
 \vspace{-20pt}
\begin{center} 
 \begin{tabular}{cc}
\begin{minipage}{0.5\hsize}
\xymatrix{
\bar{*} \ar[rd]   & \bar{*}  \\ \mbox{*} \ar[ru]^(.4){\ax_1}  &    \mbox{*} }
\end{minipage}
&
\begin{minipage}{0.4\hsize}
\xymatrix
{ & \bar{*} \ar[rd]   & \bar{*}  &  &  \\
&   \ar@{<-} `l []-<2em,0cm> `[u]+<0cm,1.5em> `[rr]+<2em,0cm> `[l] [rr] 
 \mbox{*} \ar[ru]^(.4){\ax_1}  &   \mbox{*}  \ar[rd]  &   \mbox{*} & 
 \\
& \bar{*} \ar[rd]   & \bar{*} \ar[ru]_(.6){\sigma_\textrm{cut}}  & \bar{*}  
\ar `r []-<2em,0cm> `[d]-<0cm,1.5em> `[ll]-<2em,0cm>   `[r] [ll]   & \\
&  \mbox{*} \ar[ru]^(.4){\ax_3}   &  \mbox{*}  &  }
\end{minipage}\\
Diagram of (\ref{ex3}) & Diagram of (\ref{ex2})
\end{tabular}
\end{center}
\vspace{-20pt}
\end{wrapfigure}
The GoI interpretation of $\pi_1$
is that for the indexed $\nu \in \abs{\pi_1}^J$ in Step 1,
which is defined pointwise (for $j \in J=\{ 1,2 \}$)
at $\nu_1$ and $\nu_2$, 
in which $\sigma_{\nu_1}$ 
and $\sigma_{\nu_2}$ are stipulated respectively by
$\sigma_\textrm{cut}$ and 0, where
$\sigma_\textrm{cut}$ is a symmetry $s_{(1,*),(1,\bar{*})}$
for the cut formulas while 0 is the zero morphism $0_{U^2}$
resulting by zero action 
on a symmetry $s_{(2,*),(1,\bar{*})}$: 

 \begin{flalign}
  \Ex{\sigma_{\nu_1}}{\nu_1} & =
 \TR{(\bar{*} \otimes \sigma_\textrm{cut} \otimes *) \Comp (s_{\bar{*},(1,*)}
  \otimes s_{(1, \bar{*}),*})}{
 (1,*) \,  \otimes \,  (1,\bar{*})}{\bar{*} \otimes * \,}{\,\bar{*}
  \otimes *}, 
 \label{ex11} \\ 
 & \hspace{2cm} \text{
 in which $\sigma_\textrm{cut}$ arises because of the matching $(1,*)=(1,\bar{*})$.} \nonumber  \\
 \Ex{\sigma_{\nu_2}}{\nu_2} & =  
 \TR{(\bar{*} \otimes \,  0 \,  \otimes *) \Comp (s_{\bar{*},(2,*)} \otimes
  s_{(1, \bar{*}),*} )}{(2,*) \, \otimes \, (1,\bar{*})}{\bar{*}
  \otimes * \,}{\, \bar{*} \otimes *}  \label{ex12} 
 \\
 & \hspace{3cm}
 \text{in which 0 arises because of the mismatch $(2,*) \not
  =(1,\bar{*})$.}  \nonumber
 \end{flalign}
Here (\ref{ex11}) is equivalent to (\ref{ex2}), while (\ref{ex12}) 
reduces to 0
in ${\cal C}$ by virtue of the trace axioms with zero morphisms.
The next two diagrams, for (\ref{ex11})
and (\ref{ex12}), illustrate that (\ref{ex12}) yields a zero
morphism because chasing
any arrow results in passing through 0.


\bigskip

\pagebreak

\begin{wrapfigure}[12]{r}{10cm}
\vspace{-20pt}
\begin{center} 
\begin{tabular}{cc}
\begin{minipage}{0.5\hsize}
\xymatrix@C=0.1in
{ & \bar{*} \ar[rd]   & \bar{*}  &  &  \\
&   \ar@{<-} `l []-<2.8em,0cm> `[u]+<0cm,1.5em> `[rr]+<2em,0cm> `[l] [rr] 
(1, *) \ar[ru]^(.4){\ax_1}  &    (1, *)   \ar[rd]  &  (1, *) & 
 \\
& (1, \bar{*}) \ar[rd]   & (1, \bar{*}) \ar[ru]_(.6){\sigma_\textrm{cut}}  & (1, \bar{*})  
\ar `r []-<2em,0cm> `[d]-<0cm,1.5em> `[ll]-<2.8em,0cm>   `[r] [ll]   & \\
& \mbox{*} \ar[ru]^(.4){\ax_3}   &  \mbox{*} &  }
\end{minipage} & \hspace{.5cm}
\begin{minipage}{0.5\hsize}
\xymatrix@C=0.1in
{ & \bar{*} \ar[rd]   & \bar{*}  &  &  \\
&   \ar@{<-} `l []-<2.8em,0cm> `[u]+<0cm,1.5em> `[rr]+<2em,0cm> `[l] [rr] 
(2, *) \ar[ru]^(.4){\ax_2}  &    (2, *)   \ar[rd]  &  (2, *) & 
 \\
& (1, \bar{*}) \ar[rd]   & (1, \bar{*}) \ar[ru]_(.6){0}  & (1, \bar{*})  
\ar `r []-<2em,0cm> `[d]-<0cm,1.5em> `[ll]-<2.8em,0cm>   `[r] [ll]   & \\
& \mbox{*} \ar[ru]^(.4){\ax_3}   &  \mbox{*} &  }
\end{minipage}\\
Diagram of (\ref{ex11}), $j=1$ & 
Diagram of (\ref{ex12}), $j=2$ 
\end{tabular}
\end{center}
\end{wrapfigure}
At $j=2$, \\ $\Ex{\sigma_{\nu_2}}{\nu_2}=0$,
so we delete the index 2, reducing $J$ into the singleton $\{ 1 \}$.
For index 1, 
$\Ex{\sigma_{\nu_1}}{\nu_1}$ is identical to
the symmetry $(\ref{ex3})$ of $\abs{\pi_3}$ and, hence, to
the denotational interpretation.



\bigskip





\begin{thebibliography}{1}

\bibitem[BucEhr00]{BE}
A. Bucciarelli and  T. Ehrhard,
 On phase semantics and denotational semantics in multiplicative-additive linear logic. 
 Ann. Pure Appl. Logic 102(3)  247-282 (2000)


\bibitem[BucEhr01]{BE2}
A. Bucciarelli, T. Ehrhard, 
On phase semantics and denotational semantics: the exponentials. Ann. Pure Appl. Logic 109(3): 205-241 (2001)

 \bibitem[DR95]{DR}
 V. Danos and L. Regnier,
 Proof-nets and the Hilbert space. In J.-Y. Girard, Y. Lafont, and
	 L. Regnier, eds. Advances in Linear Logic, Cambridge
	 University Press, (1995) 307--328

\bibitem[Duc09]{Duchesne}
E. Duchesne,
La Localisation en Logique : G\'{e}om\'{e}trie de
l'Interaction et S\'{e}mantique D\'{e}notationnelle. Th\`{e}se de doctorat,
Universit\'{e} Aix-Marseille II, 2009.


 \bibitem[Gir89]{Gi89}
J-Y. Girard,  Geometry of Interaction I:
Interpretation of System F, in: {\em Logic Colloquium '88 \/}, 
North-Holland, 1989, pp. 221-260.




\bibitem[Gir95]{Gi95} J-Y. Girard, Geometry of Interaction III:
Accommodating the Additives, in: {\em Advances in Linear Logic,} LNS \
\textbf{222}, CUP, 1995, 329--389.




\bibitem[Gir11]{Gi11}
J-Y. Girard,
Geometry of Interaction V: Logic in the Hyperfinite Factor,
 Theor. Comput. Sci. Vol. 412 No. 20  (2011) pp. 1860-1883

  \bibitem[HS06]{HS06} E. Haghverdi and P. Scott, A Categorical Model
for the Geometry of Interaction, Theor. Comput. Sci. 
Vol. 350 (2-3),  (2006), pp. 252-274.

\bibitem[HS11]{HS11} E. Haghverdi and P. Scott,  Geometry of Interaction
and the Dynamics of Proof Reduction:
	a tutorial, in: New Structures in Physics, {\em 
	 Springer, Lect. Notes in Phys. 
	}, (2011), 339-397




\bibitem[Ham17]{Ham}
M. Hamano, 
Geometry of Interaction for MALL via Hughes-vanGlabbeek Proof-Nets,
 ACM Transactions on Computational Logic 19 (4):1-25:38 (2018)
	



\bibitem[HamTak08]{HamTak08} M. Hamano and R. Takemura,  An Indexed System for
	Multiplicative Additive Polarized Linear Logic , Proc. of
	CSL'08, LNCS, 5213 (2008), pp. 262-277. 



\bibitem[JSV96]{JSV96}
A. Joyal,  R.  Street,  and D. Verity,
Traced Monoidal Categories.
{\em Math. Proc. Camb. Phil. Soc.}  119 (1996), pp. 447-468.

\bibitem[Sei16]{Seil}
T. Seiller, Interaction Graphs: Additives,
Ann. Pure Appl. Logic 167 (2016), pp. 95-154.







\end{thebibliography}
\end{document}